\documentclass[11pt]{article}
\usepackage{amsmath,amssymb,amsthm,mathrsfs}
\usepackage{natbib, graphicx,color}
\usepackage{multirow}
\usepackage{rotating,tabularx}
\usepackage{algorithm, algorithmic}

\newtheorem{Theorem}{Theorem}
\newtheorem{Lemma}{Lemma}
\newtheorem{Corollary}{Corollary}

\newtheorem{Proposition}{Proposition}

\newtheorem{Remark}{Remark}

\newcommand{\be}{\begin{equation}}
\newcommand{\ee}{\end{equation}}
\newcommand{\beqa}{\begin{eqnarray*}}
\newcommand{\eeqa}{\end{eqnarray*}}
\newcommand{\beqn}{\begin{eqnarray}}
\newcommand{\eeqn}{\end{eqnarray}}
\newcommand{\baa}{\begin{array}}
\newcommand{\eaa}{\end{array}}
\newcommand{\bcc}{\begin{center}}
\newcommand{\ecc}{\end{center}}
\newcommand{\btab}{\begin{tabular}}
\newcommand{\etab}{\end{tabular}}

\newcommand{\mb}{\makebox}

\newcommand{\lb}{\label}

\newcommand{\nn}{\nonumber}

\newcommand{\mc}[1]{\mathcal{#1}}

\newcommand{\mcD}{\mathcal{D}}

\newcommand{\mcM}{\mathcal{M}}

\newcommand{\tr}{{\rm tr}}

\newcommand{\argmin}{\mb{\rm argmin}}

\newcommand{\nin}{\not\in}

\newcommand{\bsigma}{\boldsymbol{\sigma}}

\newcommand{\btheta}{\boldsymbol{\theta}}
\newcommand{\bLambda}{\boldsymbol{\Lambda}}

\newcommand{\bDelta}{\boldsymbol{\Delta}}

\newcommand{\bPi}{\boldsymbol{\Pi}}

\newcommand{\bSigma}{\boldsymbol{\Sigma}}

\newcommand{\bK}{\boldsymbol{K}}
\newcommand{\bA}{\boldsymbol{A}}
\newcommand{\bB}{\boldsymbol{B}}

\newcommand{\bba}{\boldsymbol{b}}
\newcommand{\bC}{\boldsymbol{C}}
\newcommand{\bc}{\boldsymbol{c}}
\newcommand{\bD}{\boldsymbol{D}}

\newcommand{\bff}{\boldsymbol{f}}
\newcommand{\bV}{\boldsymbol{V}}
\newcommand{\bv}{\boldsymbol{v}}
\newcommand{\bG}{\boldsymbol{G}}
\newcommand{\bH}{\boldsymbol{H}}

\newcommand{\bM}{\boldsymbol{M}}

\newcommand{\bI}{\boldsymbol{I}}

\newcommand{\bS}{\boldsymbol{S}}

\newcommand{\bE}{\boldsymbol{E}}
\newcommand{\bF}{\boldsymbol{F}}

\newcommand{\bP}{\boldsymbol{P}}

\newcommand{\bQ}{\boldsymbol{Q}}
\newcommand{\bR}{\boldsymbol{R}}

\newcommand{\bU}{\boldsymbol{U}}
\newcommand{\bx}{\boldsymbol{x}}
\newcommand{\bX}{\boldsymbol{X}}

\newcommand{\bZ}{\boldsymbol{Z}}

\newcommand{\bee}{\boldsymbol{e}}
\newcommand{\bzero}{\boldsymbol{0}}
\newcommand{\bone}{\boldsymbol{1}}

\newcommand{\RR}{{\mathbb R}}

\def \nn{\nonumber}

\def\Tr{{\rm{Tr}}}

\title{A Maximum Correlation Entropy Implementation of Differential Privacy Under Linear Invariants}

\title{\bf A Maximum Entropy Implementation of Differential Privacy Under Linear Invariants}
  
\author{Ryan Lafferty and  Anindya Roy \\
   Department of Mathematics and Statistics,\\ University of Maryland Baltimore County}
\date{}

\begin{document}
\maketitle
\begin{abstract}%
Differential privacy is the standard for ensuring data privacy and is widely used in major data publications, including reporting results from the U.S. decennial census. Common implementation of differential privacy uses independent  Gaussian or Laplace noise addition to the database. However, there could be aggregate (linear) queries to the database that are excluded from the privacy budget, for example,  state totals that can not be perturbed due to constitutional mandates. Any implementation of a differential privacy is required to honor these constraints, also referred to as invariants. Under aggregation constraints, the noise vector is no longer independent and the traditional differential privacy guarantees have to be re-evaluated. We propose a high entropy differential privacy implementation that maintains the aggregation invariants with probability one or exponentially close to one and  derive the privacy guarantees for the implementation under the invariants. The theoretical proof also gives a solution to an open question about the null space of correlation matrices. Moreover, the methodology  has general use in the context of sampling from normal mixture models under linear equality constraints.\\
{\bf  keywords}: Correlation matrix; Gaussian mechanism; Null vector; Projection to convex sets; Projected gradient algorithm
\end{abstract}

\section{Introduction}\label{sec-intro}

Statistical disclosure control (SDC) is a way of protecting data privacy of large-scale data releases by agencies such as the U.S. Census Bureau or the Office of National Statistics. A popular way of ensuring SDC is through the implementation of {\it differential privacy} (\cite{DwMcNiSm2016}). One of the  attractive features of differential privacy (DP) is that it can be implemented by  simple noise addition mechanism where noise is sampled from appropriately chosen marginal distributions such as Laplace or Gaussian.

For DP implementation to large databases such as those in the decennial census, the curator is often faced with natural constraints which force the DP mechanism to bypass certain queries; \cite{Abowd20222020}. These queries, called 'invariants', and often stated in terms of aggregation constraints,  are kept outside the privacy budget. The constraints still pose significant challenges to the privacy implementation. In particular, as pointed out in \cite{Abowd20222020}, significant post-processing to the implementation algorithm is required to meet the constraints. Such post-processing may modify the privacy guarantees of the overall scheme and interpretation of the differential privacy measure. The comprehensive article \cite{Abowd20222020} describing the top-down algorithm for DP implementation in the 2020 census states "... but leave for future work a discussion of how users can best analyze the released data and how the privacy semantics—the interpretation of DP—are modified in the presence of invariants." 
One of the main goals  is to contribute to that discussion and propose a principled way of addressing the problem of aggregation constraints in DP-implementation while  having fidelity to the standard interpretation of DP.

While our approach is inspired by noise infusion techniques used in last decennial census, the idea of applying differential privacy under linear constraints can be found in other fields as well. For example, smart meter data often satisfy aggregation relationships across households, feeders, substations, or time intervals, and preserving these relationships is essential for downstream analyses such as load forecasting, demand response, and network monitoring; \cite{cardenas2026multiscale}. Similarly, wearable device data are frequently summarized over multiple temporal or spatial resolutions, requiring released data to remain coherent across aggregation levels. 

To fix ideas, we begin with a synthetic example motivated by the census exercise. Consider a particular state. We illustrate the objective based on the total count of that state, say $S$, being kept fixed in the privacy implementation process  following constitutional mandates. Formally, let $Q_{ijk}$ be the counts for the $k$th block nested in the $j$th tract in  county $i$ of the state. Then $S = \sum_{i,j,k} Q_{ijk}.$ If the block level counts are changed in the privacy implementation using the Laplace mechanism, then the perturbed values would be 
\[ {\tilde{Q}}_{ijk} = Q_{ijk} + X_{ijk}\]
where $X_{ijk}$ are the  block level noise added to the count. For the Laplace mechanism, $X_{ijk} \sim \text{Laplace}(\lambda)$ and for the Gaussian mechanism, $X_{ijk} \sim N(0, \sigma^2)$. In reality the parameters can be made to dependent on the specific geographical area, e.g. $\lambda = \lambda_{ijk}$ or $\sigma = \sigma_{ijk}$, but for illustration we choose them to be constant.  Then the state-level total using the perturbed block-level counts is 
\[ {\tilde{S}} = S  + \sum_{i,j,k} X_{ijk}.\]
To exclude the state total from the privacy budget one would require 
\[\sum_{i,j,k} X_{ijk} = 0.\]
The linear constraint on the   noise vector, $\bX = \{X_{ijk}\}_{i,j,k},$ can be represented as $a^TX = 0$ where $a = (1, 1, \ldots, 1)$ is a vector of ones. Thus, we want to add  noise with specified marginal to each block such that the state-level aggregate of  block-level noise is zero with probability one.

The constraints on the sample of block-level noise will make the noise dependent and the noise will live in a lower-dimensional noise manifold. Hence the original interpretation of DP no longer holds. To devise a mechanism that closely adheres to the original DP-implementation through a Laplace or Gaussian mechanism yet satisfies the aggregation constraints, we propose the following restrictions on the noise vector $\bX$. 

\begin{itemize}
\item 

{\bf [R1]} The marginal distribution of the noise matches the specified distributions under the unconstrained DP-implementation. 
\item 

{\bf [R2]} The noise vector must satisfy the aggregation constraints with probability close or equal to one. 

\item 

{\bf [R3]} The noise vector must have a high entropy with respect to the maximum dimension possible of the noise-manifold under [R1] and [R2].  
\end{itemize}

The reasons for the requirements are the following. Requirement [R1] on the marginal distribution is convenient  for establishing privacy guarantees at the lower level of  aggregation, similar to a  local DP formulation \cite{kasiviswanathan2011can}. Restriction [R2] is needed for the implementation to honor the invariants. However, beyond the constraints imposed by the invariants, the privacy objectives warrants that the noise vector should be "as close to being independent" as possible. Restriction [R3] ensures that by requiring the noise vector to have the maximum entropy permitted within the restrictions [R1] and [R2]. We show that it is possible to devise an implementation that satisfies [R1-R3], and has formal privacy guarantees. 
The methodology is intimately connected to the properties of correlation matrices, and we contribute some new results to that literature. In particular, we provide a partial solution to an open question related to the structure of null spaces of correlation matrices. 

It should be noted that the noise added  quantity under a standard DP implementation can be projected to the space of invariant constraints to obtain constrained noise-added quantity. It is tempting to argue that under the 'invariance to post-processing' property of DP, such a solution is sufficient. However, as illustrated in the methods section, the neighborhood structure for evaluating participation privacy under invariance is more nuanced and the projection approach fails to capture that. Moreover, the projected quantities reside on the boundary of the set, limiting the entropy of the solution.

The problem of sampling the noise vector with constraints can be cast into a much broader framework. Both types of marginal density relevant to the DP implementation, Laplace and Gaussian, belong to the family of  densities that can be represented as scale mixture of $N(0, \sigma^2)$ density.  Many of the results presented here are applicable to the  general scale-mixtures $N(0, \sigma^2)$. The scale mixtures of $N(0, \sigma^2)$ for a rich family of distributions that are frequently used in statistical modeling and inferential procedures. While sampling from normal scale mixtures is a common and well-understood, sampling from such distributions under linear constraints is a challenging problem. A special case has been studied in \cite{Geweke1998} where the marginals are normal. Researchers in applied sciences have looked into the problem of sampling under linear equality and inequality constraints where the target density is restricted to hyperplanes; \cite{Cong2017,Maatouk2022}, \cite{Hoffman1991, Journel1991}. Sampling from multivariate densities with specified marginals is a well-studied problem, particularly through copula-based methods \cite{joe1993parametric,koehler1995constructing}. Also, modeling multivariate distributions with specified marginals and correlation structure has been studied in the literature; \cite{dukic2013minimum, huber2015multivariate}, but the present situation is more involved due to the linear equality constraints.

The rest of the paper is organized as follows. In Section~2 we introduce the main methodology 
by describing a formal DP mechanism appropriate for the problem at hand and characterize  for sampling noise with desired marginals under linear constraints. In Section~3, we show that the DP implementation with the proposed noise addition does provide conventional DP guarantees. Section~4 looks at theoretical properties of the noise addition algorithm. We present some numerical illustrations in Section~5 and conclude with a discussion in Section~6.

\section{Methodology}\label{sec:meth}

\subsection{Aggregation Invariant DP (AI-DP) Mechanism}
Both Gaussian and the Laplace distribution that are used as marginals for the DP formulation are scale mixtures of Normal. Thus to maintain [R1],  given scales  $\bsigma = (\sigma_1, \ldots, \sigma_n)^T$, we could marginally generate $X_i | \sigma_i \sim N(0, \sigma_i^2).$  To impose [R2], we propose the scheme
\[  \bX | \bsigma \sim N(\bzero,  \bD\bR\bD)\]
where $\bD = Diag(\sigma_1, \ldots, \sigma_n)$
and $\bR$ is a correlation matrix such that $\bR\bD\bA^T = \bzero.$
Choosing  $\sigma_i = \sigma$ for all $i$ or $\sigma_i \stackrel{iid}{\sim} Rayleigh(\lambda)$,  gives the  Gaussian and Laplace marginals, respectively.  
Given the prevalence of the Gaussian-DP mechanisms, we focus on invariant Gaussian mechanism.
Specifically, let $\mathscr{D}$ be the entire data universe consisting of all possible databases $\mcD$  and consider an $n$-dimensional query $Q:\mathscr{D} \to \RR^n.$ Following Define the "$\bA$\text{-}invariant data universe  $\mathscr{D}_{Q, \bA, \bc}$ with respect to a linear query $Q$ satisfying invariant constraints with respect to a constraint matrix $\bA$ as  
\[ 
    \mathscr{D}_{Q, \bA, \bc} = \{ \mcD \in \mathscr{D}: \bA Q(\mcD) = \bc \}.
\]
We denote neighboring databases in the invariant data universe as $\mcD \sim \mcD^{\prime}$ if the databases differ by a single entry of $\mcD$ being swapped with a single entry of $\mcD^{\prime}$ such that invariant constraint is met in both databases. We propose to develop {\bf invariant Gaussian DP mechanism,} (IDP-Gaussian) a random mechanism $\mathscr{M}: \RR^n \to \RR^n$   such that for any  database $\mathcal{D} \in \mathscr{D}_{Q, \bA, \bc}$ we have 
$$\mathcal{M}(Q(\mcD)) = Q(\mcD) + \bX,\;\; \bX \sim N(0, \sigma^2\bR),\;\; \bR\bA^T = \bzero.$$

\noindent{\bf Formal Invariant $(\epsilon, \delta)$\text{-}DP Privacy Framework:} 
Under the invariant constraint, it is not possible to define the DP guarantees in the classical sense. The guarantees are evaluated with respect to data universe where each database satisfies the invariant constraint.

In this section we look at some of the theoretical properties of the proposed algorithm for implementing privacy under aggregation. First we show that the mechanism provides formal privacy guarantees in the form of $(\epsilon, \delta)-DP$ for the Gaussian noise addition case. Under the invariant constraint, it is not possible to define the DP guarantees in the classical sense. The guarantees are evaluated with respect to data universe where each database satisfies the invariant constraint with respect to a given query, $Q$. Specifically, let $\mathscr{D}$ be the entire data universe consisting of all possible databases $\mcD$. Define the "$Q-$invariant data universe" $\mathscr{D}_{Q, \bA, \bc}$ with respect to a linear query $Q$ satisfying invariant constraints with respect to a constraint matrix $\bA$ as  
\begin{equation}
    \mathscr{D}_{Q, \bA, \bc} = \{ \mcD \in \mathscr{D}: \bA Q(\mcD) = \bc \}.
    \label{eq:invariant_data_universe}
\end{equation}
The definition of the "$Q-$invariant data universe" 
is analogous to the invariant data universe described in \cite{Bailie2026Refreshment}, although there the authors treat a more general case. Unlike the traditional  neighboring structure where a single entry in the database is perturbed, under multiple aggregation constraints the neighboring databases need to have multiple entries perturbed for the invariant criterion to hold. 
For example, if there are $k$ linearly independent aggregation constraints, then the two neighboring databases may differ by $k$ entries with their aggregates  remaining the same.  
We denote neighboring databases in the invariant data universe as $\mcD \sim \mcD^{\prime}$ if the databases differ by a single entry of $\mcD$ being swapped with a single entry of $\mcD^{\prime}$ such that invariant constraint is met in both  databases. For example, if the state level totals are left invariant then two neighboring databases could be one where a single record of one state is swapped with one from another which leaves the state totals invariant. A privacy mechanism providing formal privacy guarantees for databases in the invariant data universe will be called an {\it invariant privacy mechanism}. We present the results for the case where privacy is achieved through correlated Gaussian noise addition. 

\subsection{$(\epsilon, \delta)$-Differential Privacy for an Invariant Gaussian Mechanism }

We modify the argument in appendix A.1 of \cite{dwork2014algorithmic} to show that $(\epsilon, \delta)$ differential privacy still holds under the invariant data universe. Specifically, the noise vector obtained as the Gaussian sample with the covariance matrix $\bSigma = \sigma^2\bR$ provides an $(\epsilon, \delta)$ DP guarantees.

\begin{Theorem}
\label{thm:gaussianDP}
Let $\mcD$ be a given database in the invariant data universe $\mathscr{D}_{Q,\bA, \bc}$ for a linear query $Q$ and for some $k \times n$ constraint matrix $\bA$  of rank $k$ and some fixed vector $\bc$. Also  let $\bM$ be an $(n-k)\times n$ semi-orthogonal matrix whose rows form an orthonormal basis for $ker(\bA)$, the null space of $\bA$.  Assume there exists a positive semi-definite matrix $\mathbf \Sigma$ of rank $(n-k)$  such that $\mathbf{\Sigma}\mathbf{A}^T = \bzero.$
Define a privacy mechanism $$\mathcal{M}(Q(\mcD)) = Q(\mcD) + \bX$$ where $\mathbf{X} \sim N(0,\boldsymbol\Sigma)$. Suppose $\boldsymbol{\Sigma} = \sigma^2 \bR$, where $\sigma^2 = \mathbf{\Sigma}_{1,1}$ and $\mathbf{R} = \frac{1}{\sigma^2}\mathbf{\Sigma}$ (which may not necessarily be a correlation matrix). Then $\mathcal{M}(Q(\mcD)) \in Q(\mathscr{D}_{Q, \bA, \bc}).$
Let for $a>0, \epsilon > 0$ $t_{a, \epsilon}(\sigma) = \frac{\epsilon}{\sqrt{a}} - \frac{\sqrt{a}}{2\sigma}$ and for any $t > 0$, $\kappa(t) = t^{-1}(\phi(t))$ where $\phi$ is the pdf of $N(0,1).$ Then $t_{a, \epsilon}$ and $\kappa$ are strictly monotone functions for each $a$ and $\epsilon$. Let 
\[  a_M = \sup\{\bv^T (\bM\bR\bM^T)^{-1} \bv: \bv = \bM Q(\mcD) - \bM Q(\mcD^{\prime}), \mcD \sim \mcD^{\prime}, \mcD, \mcD^{\prime} \in \mathscr{D}_{Q, \bA, \bc}\}\]
Then for any pair $(\epsilon, \delta)$ 
if $\sigma > t_{a_M, \epsilon}^{-1}(\kappa^{-1}(\delta)),$ we have 
$$P(\mathcal{M}(Q(\mcD)) \in S) \le e^\epsilon P(\mathcal{M}(Q(\mcD^{\prime})) \in S) + \delta$$ for any database pair  $\mcD \sim \mcD^{\prime}$ in $\mathscr{D}_{Q, \bA, \bc}.$ 
\end{Theorem}

\section{Implementation of AI-DP}

\subsection{Marginals in DP Mechanism: Scale Mixtures of Normal}

The marginals that are relevant for DP implementation, that provides  differential privacy guarantees at individual level or lower-level aggregation data, are Laplace and Gaussian. Both densities belong to the family of scale mixtures of normal. 
The Laplace distribution is equivalent to a scale mixture where scales are independently drawn from the Rayleigh distribution.  Specifically, 
\[
X \sim \text{Laplace}(\lambda) \mbox{  iff  } X | \sigma \sim N(0, \sigma^2),\; \sigma^2 \sim \text{Exp}\left(\frac{1}{2\lambda^2}\right).
\]
where $\text{Exp}(\eta)$ is the Exponential density with mean $\eta^{-1}.$
Thus, the scales have a distribution that is the same as the square root of an $\text{Exp}(\frac{1}{2\lambda^2})$ random variable i.e. a Rayleigh distribution with parameter $ \lambda$. The normal distribution with a given variance can also be thought of as a scale mixture where the scales are drawn from the Dirac measure at the given value of the standard deviation.
\[  X \sim N(0, \sigma^2) \mbox{  iff  } X|\tau \sim N(0, \tau^2), \mbox{ where } \tau^2 = \sigma^2 {\mbox{ with probability one.}}\]
Thus, we will describe the noise sampling algorithm in the context of Gaussian scale mixtures. First we formally define the family of densities.

A common yet large class of  densities is the class of scale mixtures of the normal distribution. The mixture density belongs to the class
\be
\mc{C}_{\mcM} = \{g_h: g_{h}(x) =  \int_{\RR} \frac{1}{\sqrt{2\pi}\sigma}e^{-\frac{x^2}{2\sigma^2}} h_{\theta}(\sigma) d\sigma;  \int_{\RR_+} h_{\theta}(\sigma) d\sigma = 1, h_{\theta}(\sigma) \geq 0.\}
\lb{eq:scale_mix}
\ee
where the mixing density $h_{\theta}(\sigma)$ is a pdf on $\RR_+$, possibly depending on some parameters $\btheta = (\theta_1, \ldots, \theta_p).$
The mixture densities are all unimodal, symmetric densities, a common shape for a noise density.
We will assume that the specified marginal  densities of the desired noise sample in privacy implementation are from the normal scale mixture class.

\subsection{Multivariate Sample With Optimized Correlation}

For the DP implementation, we want to sample a vector of noise $\bX$ that satisfies an aggregation constraint of the form $\bA\bX = \bzero$ and is such that the marginal densities for the components of $\bX$ are from the scale-mixture class $\mc{C}_{\mcM}$ given in \eqref{eq:scale_mix}. To sample $X_i \sim g_h(x)$, one can sample $\sigma_i \sim h_{\btheta_i}$ and then sample $X_i \sim N(0, \sigma^2_i), i = 1, \ldots,n.$  However, we want the $X_i$'s to satisfy the constraint $\bA\bX = \bzero.$ This necessarily makes the $X_i$ dependent. 
One way to solve this problem is to sample $\bX$ from a multivariate normal with marginal distributions equal to $N(0, \sigma^2_i).$ Specifically, we sample $\bX \sim N_n(\bzero, \bSigma)$ where $\bSigma = \bD\bR\bD$ with $\bD = \text{diag}(\sigma_1, \ldots, \sigma_n)$ as the diagonal matrix with diagonals equal to the marginal standard deviations. To satisfy the aggregation constraint with probability one we need $\text{Var}(\bA\bX )= \bA\bSigma\bA^T = \bA\bD\bR\bD\bA^T = \bzero.$ This is equivalent to the problem of finding a correlation matrix $\bR$ such that 
\be
\bR\bB = \bzero,
\lb{eq:R_constraint}
\ee
where $\bB = \bD\bA^T.$  In the case of Gaussian noise addition with a constant noise variance $\sigma^2,$ the scale matrix is $\bD = \sigma^2 \bI,$ and hence $\bB = \bA$ in \eqref{eq:R_constraint}.

To obtain a solution to \eqref{eq:R_constraint} in terms of a correlation matrix, it is helpful to study the geometry of the feasible set of solutions. 
Let $\mathscr S^n$ denote the space of real symmetric $n\times n$
matrices, and 
let 
\[
\mathscr C^n
=
\left\{
R\in\mathscr S^n:
R\succeq 0,\quad \operatorname{diag}(R)=\mathbf 1_n
\right\}.
\]
be the set of all $n\times n$ correlation matrices. This set is closed, bounded and convex. It is  the intersection of  two convex sets defined in the following.  Consider the convex sets: 
\beqn
    \mc{C}_1 := \mathscr{S}^n_+ &=& \{ \bR \in \mathscr{S}^n :  \bR \succeq 0 \}  \nn \\
    \mc{C}_2 := \mathscr{U}^n &=& \{ \bR \in \mathscr{S}^n : r_{ii} = 1, i = 1, \ldots,n\},
\lb{eq:C1C2}
\eeqn
where for two matrices $\bE, \bF,$, $\bE \succeq \bF$ means the Loewner ordering, i,e, $\bE - \bF$ is positive semidefinite (psd). The first convex set $\mc{C}_1$ is the cone of all $n\times n$ symmetric psd matrices. The second convex set $\mc{C}_2$ is an affine subspace of $\mathscr{S}^n.$ Then $\mathscr{C}^n = \mc{C}_1 \bigcap \mc{C}_2.$

The solution to the constraint problem 
\eqref{eq:R_constraint} is in the intersection of $\mathscr{C}^n$ and the following subspace of $\mathscr{S}^n$
\be
\mc{C}_3 := \mathscr{S}_{\bB} = \{ \bR \in \mathscr{S}^n : \bR\bB = \bzero. \}
\lb{eq:C3}
\ee
Thus, the solution we seek is a correlation matrix $\bR$ in the intersection of the three convex subsets $\mc{C}_1, \mc{C}_2$ and $\mc{C}_3$ of a Hilbert space $\mathscr{S}^n.$
We have an example of a convex feasibility problem \cite{censorzaknoon2018} or CFP which asks one to find a member of the set $\mc{C} = \cap_{i=1}^J \mc{C}_i$ where $\mc{C}_i$ are all convex sets. In the present case, $J = 3$.  To obtain a solution in the intersection of the convex sets we can employ the cyclic {\it Projection Onto Convex sets} (POCS) algorithm \cite{bauschke1997method} which is a generalization of the popular {\it alternating projection} \cite{von1949rings, von1951functional}. When the individual convex sets $\mc{C}_i$  are 'simple' in the sense that projections to $\mc{C}_i$ are easily computable, then the CFP of finding a  point in $\mc{C} = \cap_{i=1}^J \mc{C}_i$ can be solved efficiently using iterative projection onto the convex sets by cycling through the individual projections. If $P_{\mc{C}}$ denotes the projection to the intersection $\mc{C}$, and $P_{\mc{C}_i}$ are the individual projection operators, then $P_{\mc{C}} = \lim_{k\to\infty} (\Pi_{i=1}^J P_{\mc{C}_i})^k$ in the sense that $\| (\Pi_{i=1}^J P_{\mc{C}_i})^k(\bR) - P_{\mc{C}}(\bR)\| \to 0, $ as $k \to \infty$ provided $\mc{C}$ is nonempty.

The success of the cyclic POCS algorithm on having individual convex sets $\mc{C}_i$ where the corresponding projections $P_{\mc{C}_i}$ are easy to compute, preferably given by closed form solution. The following result establishes the closed form solution for each of the three projections involved in the present application of the cyclic POCS. 

    \begin{Proposition}
Let $\bR \in \mc{S}^n$ with spectral decomposition $\bR = \bP\bLambda\bP^T$  and let $\bR_+ = \bP\bLambda_+ \bP^T$  where for a matrix $\bG$, $\bG_+$ denotes the matrix $\bG$ with all negative entries set to zero. Also let $diag(\bR)$ be the diagonal matrix with diagonal entries equal to that of $\bR$ and $\bI$ be the $n\times n$ identity matrix. Let $P_{\mc{C}_i}, i = 1,2,3$ be the projection operators to the convex sets $\mc{C}_i, i=1,2,3$, respectively, defined in \eqref{eq:C1C2} and \eqref{eq:C3}. Then 
\beqa
P_{\mc{C}_1}(\bR) &=&  \bR_+,\\
P_{\mc{C}_2}(\bR) &=& \bR - \text{diag}(\bR) + \bI, \\
P_{\mc{C}_3}(\bR)  &=& \bPi\bR\bPi,
\eeqa
where  $\bPi$ is the orthogonal projection matrix to $(n-k)$ dimensional orthogonal complement of the column space of $\bB$.
\lb{prop:projection_to_sb}
\end{Proposition}

If $\mc{C} \neq \emptyset$ then the CFP is called consistent, otherwise it is inconsistent. For a consistent CFP, the sequence of cyclic projections (in any given order) would converge to projection for $\mc{C}$ \cite{bauschke1997method}. For inconsistent CFP, the sequence does not not converge but has limiting cycle, say $\bR_1 \to \bR_2 \to \bR_3 \to \bR_1$ where each entry in the cycle is the projection of the previous entry.   
The feasibility problem is closely related to the "nearest correlation matrix" problem which has been extensively studied in \cite{Higham1989}, \cite{Higham2002b}, \cite{Higham2016}, and \cite{Higham2002a}. 
In the next section, we make use of the "nearest correlation" algorithm based on the modified Dykstra algorithm \cite{BoyleDykstra1986}. The algorithm is available in MATLAB. For the current set up when the problem is not feasible, a final step for finding the nearest correlation to the projection to $P_{\mc{C}_3}$ is used after the solution reaches the limit cycle. 

Once the correlation matrix $\bR$ has been obtained, the noise vector $\bX$ is generated from the multivariate Gaussian distribution $N(\bzero, \bSigma)$ where $\bSigma = \bD\bR\bD$ with $\bD = diag(\sigma_1, \ldots, \sigma_n)$ is a  a diagonal matrix of scales generated from the mixing density $h_{\sigma}$ in the scale mixture representation \eqref{eq:scale_mix} of the marginal noise density. 

Algorithm~1 provides a pseudo-code for using POCS to obtained the desired correlation matrix and sample the noise vector.  
\begin{algorithm}
\caption{Correlation Sampling Using POCS}
\lb{alg:sampling}
\begin{algorithmic}
\STATE{{\bf INPUT}: $\btheta_1, \ldots, \btheta_n,$  $\bA$, $\bR_{init}$, $TOL$,  $MAX.CYCLE$}
\STATE{
Generate $\bsigma = (\sigma_1, \ldots, \sigma_n) \sim \prod_{i=1}^n h_{\theta_i}(\sigma_i).$ }
\STATE{
$\bD =  \text{diag}(\sigma_1, \ldots,\sigma_n), \quad \bB = \bA\bD$
}
\STATE{ $ \bR_3^0 = \bR_{init}$, $E = TOL + 1$, $t  = 1$}
\STATE{{\bf WHILE} {$E > TOL \mbox{ OR }   t < MAX.CYCLE$}\\
\quad \quad    $\bR_1^t = P_{\mc{C}_1}(\bR_3^{t-1})$ \\
\quad \quad $\bR_2^t = P_{\mc{C}_2}(\bR_1^{t})$ \\ 
\quad \quad $\bR_3^t = P_{\mc{C}_3}(\bR_2^{t})$\\
\quad \quad     $E = \max\{ \|\bR_1^t - \bR_3^{t-1}\|_F,\|\bR_2^t - \bR_1^t\|_F,\|\bR_3^t - \bR_2^t\|_F\} $ \\
\quad \quad t = t+1\\
{\bf END WHILE}
}
\STATE{ $\bR = P_{\mc{C}_1 \bigcap \mc{C}_2} (\bR^t_3)$  \quad "nearest correlation"}
\STATE{{\bf RETURN} $\bX = (X_1, \ldots, X_n)^T \sim N(\bzero, \bD\bR_{1}\bD)$.
}
\end{algorithmic}
\end{algorithm}

\subsection{Max-entropy Solution}

The POCS solution provides a correlation matrix in the intersection $\mc{C}.$ The solution does not have any additional properties beyond being a point in $\mc{C}.$
In order to ensure that requirement [R3] is met for the POCS solution obtained as the projection to the intersection of the convex sets, one can penalize the solution  for loss of entropy in the sampling model. Specifically, we want to ensure that the constrained Gaussian sampling model obtained under the correlation matrix solution has maximal entropy among all feasible solutions in terms of the Lebesgue density of the lower dimensional subspace it is concentrated on. There is no unique definition for the entropy of a singular distribution, but researchers have looked at some  common-sense possibilities \cite{pichler2014entropy}.
Typically, a single subspace can be chosen and the entropy is defined with respect to the Lebesgue measure for that subspace. To define the entropy with respect to a subspace of concentration, we use the following approach. Let $\bX \sim N(\bzero, \bR)$ where $\bR\bB = \bzero.$ Define a nonsingular transformation of $\bX$ as
\[ \bZ = \begin{bmatrix} \bZ_1 \\  \bZ_2 \end{bmatrix} =  \begin{bmatrix} \bB^T \\ \bK^T \end{bmatrix} \bX  = \bM^T \bX \]
where $\bM = [\bB : \bK]$ is nonsingular. There are no unique choices for $\bK$, but we can use the matrix $\bK$ defined in Proposition~\ref{prop:projection_to_sb} whose columns form an orthonormal basis of ${\rm{Col}}(\bB)^{\perp}$. With this choice, $\bZ_1$ is a constant under the constraint as $\text{Var}(\bZ_1) = \bB^T\bR\bB = \bzero$ and  the differential entropy of the lower-dimensional distribution of $\bZ_2$, is evaluated along a fixed subspace that is in the orthogonal to the direction of degeneracy, i.e. $\bB^T\bK = \bzero.$ The differential entropy of $\bZ_2$ will be $\log\det (\bK^T \bR \bK)$ (up to some additive constant. )

To facilitate the penalized approach, it is convenient to rewrite the CFP as an optimization problem. Note that if $\bR \in \mathscr{C}^n$ then $\|\bR^{1/2}\bB\|_F = 0$ where $\bR^{1/2}$ is a square root of $\bR.$ Then the original problem of finding $\bR \in \mathscr{C}^n$ can be rewritten as 
\[     \bR =  \min_{\bR \in \mathscr{C}^n} \Tr( \bB^T\bR\bB). \]
The advantage of the alternative formulation is immediate, as 
an optimal choice of $\bR$, satisfying requirements [R1-R3], will be then a solution to the following penalized problem 
\begin{equation}
\bR_{opt} = \min_{\bR \in \mathscr{C}^n}[\Tr( \bB^T\bR\bB)   - \tau \log\det(\bK^T \bR \bK)] 
\lb{eq:corr_opt}
\end{equation}
The tuning parameter $\tau$ can be chosen using standard cross-validation method but since the primary focus is on meeting the aggregation constraint; hence the value of $\tau$ is typically small. The optimization in \eqref{eq:corr_opt} can be solved using a projected gradient algorithm \cite{calamai1987projected, polyak2021introduction, ruszczynski2011nonlinear} where the following steps can be iterated until convergence:\\
\noindent{\bf Iteration steps:}
\begin{itemize}
    \item Step in the direction of greatest descent
    \item Use nearest correlation to project to $\mathscr{C}^n$
\end{itemize}
Thus, the iteration over the solution sequence is 
\[  \bR^{t+1} = P_{\mc{C}_1 \bigcap \mc{C}_2}( \bR^t - \alpha_t \nabla f(\bR^t)), \;\;t = 1,2,\ldots\]
where $f(\bR)  = \Tr( \bB^T\bR\bB)   - \tau \log\det(\bK^T \bR \bK)$, $\alpha_k$ are the step sizes which could be chosen to be a fixed value $\alpha_k = \alpha.$ Once a solution $\bR$ is obtained, the noise sample is generated from $N(\bzero, \bD\bR\bD)$ as before. The pseudo-code for the projected gradient algorithm is given in Algorithm~\ref{alg:maxent_sampling}.

\begin{algorithm}
\caption{Max-entropy sampling under specified marginal and linear constraints}
\lb{alg:maxent_sampling}
\begin{algorithmic}
\STATE{{\bf INPUT}: $(\btheta_1, \ldots, \btheta_n),$ $\bA$}
\STATE{
Generate $\bsigma = (\sigma_1, \ldots, \sigma_n) \sim \prod_{i=1}^n h_{\btheta_i}(\sigma_i).$ }
\STATE{
$\bD =  \text{diag}(\sigma_1, \ldots,\sigma_n), \quad \bB = \bA\bD$
}
\STATE{$SVD(\bA) = \bP\bDelta\bQ; \quad \bQ = [\bQ_1^t : \bQ_2^T]^t, \quad \bK = \bQ_2$}
\STATE{Use projected gradient to obtain optimal correlation matrix as 
\[
\bR_{opt} = \underset{\bR \in \mathscr{C}_n} \argmin\; [\Tr(\bB^T\bR\bB) - \tau \log \det (\bK^T\bR\bK)]
\]
}
\STATE{
{\bf RETURN} $\bX = (X_1, \ldots, X_n) \sim N(\bzero, \bD\bR_{opt}\bD^T)$.
}
\end{algorithmic}
\end{algorithm}

\section{Existence of Max-Entropy Solutions}

The algorithms presented here are all related to properties of correlation matrices and existence of desired subclasses of correlation matrices. In this section we explore further theoretical properties of the proposed algorithm.

\subsection{Necessary and Sufficient Condition for Existence of Solution}

In most applications, the single constraint case will be sufficient for the implementation of the proposed Laplace mechanism. This is because, if the invariants apply to disjoint part of the data vector, then the privacy implementation can be done in parallel over the disjoint partition and hence reduce the multiple constraint situation to several single constraint problems. For example, the constraint that the state-level total is left invariant in the noise addition scheme can be implemented state-by-state. 

In the single constraint vector cases there is an explicit characterization of vectors for which a correlation matrix exists such that the vector is in the null space of the correlation. matrix.
The case when there are more than one constraint vectors is more complicated and an explicit description of the type of constraints vectors that simultaneously belong to the null space of a single correlation matrix is only partially available \cite{saunderson2012diagonal}. 
We develop a characterization for the multiple constraint case that reduces to that of the single constraint when there is only one constraint. Also, this provides an alternative proof of the single constraint case. 

\begin{Theorem}
\label{thm:multiple_constraint}
Let $B\in\mathbb R^{n\times k}$ have rank $k$, where
$1\leq k<n$. Choose a permutation matrix $P$ such that
$PB= [B_1^\top : B_2^\top ]^\top,$
where $B_1\in\mathbb R^{k\times k}$ is nonsingular. Set
$F=B_2B_1^{-1} =[\,f_1:\cdots:f_k\,]\in\mathbb R^{(n-k)\times k}.
$
Let $e_i$ denote the $i$th coordinate vector in $\mathbb R^d$,
and define
\beqa
A_i & = & e_ie_i^\top,\qquad i=1,\ldots,(n-k),\\
A_{d+j} & = & f_jf_j^\top,\qquad j=1,\ldots,k.
\eeqa
Then there exists $R\in\mathscr C^n$ satisfying $RB=0$ if and
only if the following semidefinite program is feasible:
\begin{equation}
\label{eq:nullspace-sdp}
\begin{aligned}
\underset{C\in\mathbb S^d}{\operatorname{minimize}}
&\quad 0\\
\text{subject to}
&\quad \langle A_i,C\rangle=1,
       \qquad i=1,\ldots,n,\\
&\quad C\succeq 0,
\end{aligned}
\end{equation}
where $\langle U,V\rangle=\operatorname{tr}(UV)$. Moreover, the mapping
\begin{equation}
\label{eq:nullspace-reconstruction}
C\longmapsto
R(C):=
P^\top
\begin{pmatrix}
F^\top CF & -F^\top C\\
-CF       & C
\end{pmatrix}
P
\end{equation}
is a bijection between the feasible set of
\eqref{eq:nullspace-sdp} and the set  $\mathscr S_{\bB}$  defined in \eqref{eq:C3}.
\end{Theorem}

\begin{Remark}
\label{rem:maxdet-choice}
Theorem~\ref{thm:multiple_constraint} requires only that
$B_1$ be nonsingular. If $B_1$ is additionally chosen to
maximize the absolute determinant among all submatrices
formed by selecting $k$ rows of $B$, then
\[
|f_{ij}|\leq1,
\qquad i=1,\ldots,(n-k),\quad j=1,\ldots,k.
\]

Indeed, let $b_{1,\ell}$ denote row $\ell$ of $B_1$ and
let $b_{2,i}$ denote row $i$ of $B_2$. Since $B_2=FB_1$,
\[
b_{2,i}
=
\sum_{\ell=1}^k f_{i\ell}b_{1,\ell}.
\]
Let $B_1^{(i,j)}$ be obtained from $B_1$ by replacing its
$j$th row with $b_{2,i}$. Multilinearity of the determinant,
together with its vanishing when two rows coincide, gives
\[
\det\bigl(B_1^{(i,j)}\bigr)
=
f_{ij}\det(B_1).
\]
Up to a permutation of its rows, $B_1^{(i,j)}$ is a
$k$-row submatrix of $B$. Therefore, maximality of
$|\det(B_1)|$ implies
\[
\left|\det\bigl(B_1^{(i,j)}\bigr)\right|
\leq |\det(B_1)|.
\]
Dividing by $|\det(B_1)|>0$ yields $|f_{ij}|\leq1$.
Equality may occur, so strict inequality does not follow
in general.
\end{Remark}


The necessary and sufficient condition for existence of a correlation matrix such that the constraints belong to the null space of the matrix reduces to a compact form for the case when $k = 1$. The relevant result is that for a given vector $\bba \in \RR^n$, there exists $\bR \in \mathscr{C}^n$ such that $\bR\bba = \bzero$ iff the vector $\bba$ is balanced; see \cite{Delorme1993, Barrett2003, saunderson2012diagonal}. In $\RR^n$, a vector is called balanced if it belongs to the set 
\be
\mathscr{B}^n = \{\bba \in \RR^n: |b_i| \leq \sum_{i\ne j} |b_j|, i = 1, \ldots, n\}.
\lb{eq:balanced}
\ee
This condition follows immediately from Theorem~\ref{thm:multiple_constraint}.  The  following corollary states the result for the single constraint case. 
\begin{Corollary}
Let $\bba   \in \RR^n$. There exists $\bR \in \mathscr{C}^n$ with $\bR\bba = \bzero $ iff $\bba \in \mathscr{B}^n.$
\lb{cor:single_constraint}
\end{Corollary}

\subsection{Max-Rank solution}

The set of correlation matrices $\mathscr{C}_n$ is a compact convex body known as the elliptope. We will discuss some of the geometric aspects of this set and use them to locate the solutions of highest rank in $\mathcal{V}_\mathbf{B}$. Consider the problem of finding a correlation matrix $\mathbf{R} \in \mathscr{C}_n$ such that $\mathbf{R}\mathbf{B} = \mathbf{0}$ for a given $n\times k$ constraint matrix $\mathbf{B}$. Let $\mathcal{V}_\mathbf{B}$ be the set of solutions. 
We are interested in finding a solution whose rank is maximal among the matrices in $\mathcal{V}_\mathbf{B}$. 

First, observe that a correlation matrix $\mathbf{R}$ solves $\mathbf{R}\mathbf{B} = \mathbf{0}$ if and only if $\ker \mathbf{R} \supseteq \text{col } \mathbf{B}$. Given any subspace $\mathbf{V}$ of $\mathbb{R}^n$ we can form $\mathbf{B}$ whose columns are a basis of $V$. By Theorem 2.7 of \cite{Laurent1996OnTF}, the faces of the elliptope $\mathscr{C}_n$ are precisely the sets of the form $$\mathbf{F}_\mathbf{V} = \{\mathbf{R} \in \mathscr{C}_n | \ker \mathbf{R} \supseteq \mathbf{V} \},$$ for a subspace $\mathbf{V}$ of $\mathbb{R}^n$. Hence, every face of the elliptope is of the form $\mathbf{F} = \mathcal{V}_\mathbf{B}$ and conversely, every set $\mathcal{V}_\mathbf{B}$ is some (possibly empty) face of $\mathscr{C}_n$.


\begin{Theorem}
    A matrix $\mathbf{R}$ has maximal rank among the elements of $\mathcal{V}_\mathbf{B}$ if and only if it belongs to the relative interior of $\mathcal{V}_\mathbf{B}$. In other words, the set of matrices in $\mathcal{V}_\mathbf{B}$ having maximal rank is exactly the set $\text{ri}(\mathcal{V}_\mathbf{B})$. 
Moreover, if  $\mathbf{R} \in \mathcal{V}_\mathbf{B}$, then for any $\epsilon>0$ there exists a matrix $\mathbf{U}$ with $||\mathbf{U}|| < \epsilon$ such that $(\mathbf{R} + \mathbf{U})\mathbf{B} = \mathbf{0}$, $\mathbf{R}+\mathbf{U} \in \mathscr{C}_n$ and $\mathbf{R} +\mathbf{U}$ has maximal rank among solutions to $\mathbf{X}\mathbf{B} = \mathbf{0}$.
\label{thm:maxrank}
\end{Theorem}

\subsection{Probability of violation}

The most common linear query would be a simple aggregation where the constraint vector is a vector of ones. The vector of ones is a balanced vector and hence meets the necessary and sufficient condition for existence of a correlation matrix which annihilates the vector. In face in the case of the vector of ones, the max-rank, max entropy corelation matrix ${\bar{\bR}}$ can be hard-coded and would not need any optimization. The correlation for this special and most relevant case is given in the next section. 
Even when there are multiple queries, the columns of the constraint matrix  will, in most cases, have disjoint support, meaning the location of nonzero entries of the columns will be disjoint sets.  Hence jointly they will satisfy  the necessary and sufficient condition for existence of a single correlation matrix that annihilates all of them simultaneously. 

However, there maybe cases then the query vector is not a constant vector. In such a case, it is reasonable to determine the likelihood of having a query vector which does not belong to the balanced set, as defined in \eqref{eq:balanced}. While the queries are mostly fixed vector, in situation, such as survey sampling, the vector could be a vector of weights randomly generated based on random sample sizes. The most likely scenario when the vector won't be balanced, that is one entry will dominate the sum of the rest (in absolute value) is then the magnitudes of the entries follow a heavy-tailed distribution. The condition that the constraints are not met in a single constraint case has deeper connection to statistical process control  in many engineering applications. The event that the constraints are violated can be connected to a process being out of control in the peak-to-average (PTA) ratio measure. The tail event of the out-of-control PTA  used in the monitoring \cite{arendarczyk2018joint,balakrishnan2013scale, qeadan2012joint, kozubowski2010distributions, haas1992maximum, breiman1965some, morrison1965some, arov1960extreme, darling1952influence} is precisely the constraint violation event in the current set up. When the magnitude of the entries in the constraint vector arise from a Pareto Type-II distribution, we can provide the following tail bound.
\begin{Proposition}
Let $\bba \in \RR^n$ be a constraint vector whose entries follow a Lomax or a Pareto Type-II distribution. 
Then $P(\bba \nin \mathscr{B}^n) \leq (2n) (0.5)^{n}$
\label{prop:lomax_violation}
\end{Proposition}
The bound is conservative and can be sharpened considerably using specific distributional properties of the entries of the constraint vector. Moreover, if the distribution of the constraint vector entries can be represented as a mixture where, under each component of the mixture, the probability of generating an unbalanced vector is exponentially small, the overall probability of generating an unbalanced vector under the mixture will also be exponentially small.

\section{Numerical illustration}

\subsection{ Simple Aggregation:}

Consider a simplified example (motivated by the Census application) where we impose that block level counts $Q_i$ (for a certain state) must sum to the state-level total $S$. This implies that the added noise $X_i$ must satisfy the constraint $\bba^T \mathbf{X} = \mathbf{0}$ where $\mathbf{b}$ is the vector of $n$ ones. Obviously $\mathbf{b}$ is balanced. In this simple case, it is possible to give an exact description of a correlation matrix $\mathbf{R}$ satisfying $\mathbf{R}\mathbf{b} = \mathbf{0}$. Take ${\bar{\bR}}_n $ to be the matrix with $1$ on the diagonal, and $-\frac{1}{n-1}$ elsewhere, i.e 
\begin{equation}
    \label{eq:intraclass}
{\bar{\bR}}_n = \left(1 + \frac{1}{n-1}\right)\bI_n - \frac{1}{n-1}\bone_n\bone^T_n 
\end{equation}
where $\bI_n$ is the $n\times n$ identity matrix and $\bone_n $ is the $n\times 1$ vector of ones. Note that ${\bar{\bR}}_n $  is the unique $n\times n$ rank-deficient {\it intraclass correlation matrix} of rank $(n-1).$ Given that ${\bar{\bR}}_n $ is the only max rank solution, for all reasonably small values of the penalty $\tau$, ${\bar{\bR}}_n $ is also the max-entropy solution to \eqref{eq:corr_opt}.

More generally, suppose there are multiple aggregation constraints on a disjoint set of queries, such as if one wishes to have the county totals as invariants for several counties in a state. In that case, the constraint matrix would be of the block-diagonal form $$\mathbf{B} = \begin{bmatrix}
    \mathbf{1}_{n_1} & \\ & \mathbf{1}_{n_2} && \\& & \ddots &&\\ &&&\mathbf{1}_{n_N}
\end{bmatrix}. $$ 

Then it is easy to see that the block-diagonal correlation matrix $$\bR = \begin{bmatrix}{\bar{\bR}}_{n_1} & \\ & \ddots & \\ & & {\bar{\bR}}_{n_N}\end{bmatrix}$$ annihilates $\mathbf{B}$ and is the max-rank, max-entropy solution to \eqref{eq:corr_opt}. 

Let $\mathbf{M}$ be the matrix formed by taking the $n \times n$ identity matrix $\mathbf{I}_n$ with $n = \sum_{i=1} n_i$ and deleting its rows at indices $n_1, n_1 + n_2 ,..., \sum_{i=1} n_i $. One can check that the rows of this matrix extend those of $\mathbf{B}^T$ to a basis. The transformation $\mathbf{R} \rightarrow \mathbf{M} \mathbf{R} \mathbf{M}^T$ will have the effect of deleting the final row and column from each of the constituent intraclass correlation matrices $\mathbf{\bar{R}}_n$. After some algebraic manipulations we can express the inverse 

$$(\mathbf{M}\mathbf{R}\mathbf{M}^T)^{-1} = \begin{bmatrix}{\mathbf{H}}_{n_1} & \\ & \ddots & \\ & & \mathbf{H}_{n_N}\end{bmatrix},$$ where $\mathbf{H}_i = \frac{n_i-1}{n_i}(\bI_{n_i -1} + \bone_{n_i -1}\bone^T_{n_i - 1} )$. Furthermore if we let $\tilde{\mathbf{H}}_i$ represent $\mathbf{H}_i$ with a row and column of zeros added at the end, we can write:
$$\mathbf{M}^T(\mathbf{M}\mathbf{R}\mathbf{M}^T)^{-1}\mathbf{M} = \begin{bmatrix}{\tilde{\mathbf{H}}}_{n_1} & \\ & \ddots & \\ & & \tilde{\mathbf{H}}_{n_N}\end{bmatrix}.$$

In this case, we compute a bound for $a_M$, which is the analog of sensitivity in the correlated Gaussian mechanism. It is straightforward to see that
$$a_M \ge \sum_{i=1}^N (Q_i(\mathcal{D})-Q_i(\mathcal{D}'))^T\tilde{\mathbf{H}}_{n_i}(Q_i(\mathcal{D})-Q_i(\mathcal{D}')),$$
where here $Q_i(\mathcal{D})$ represents the vector of block-level counts for the $i$th county given the data set $\mathcal{D}$, and where $\mathcal{D}$ and $\mathcal{D}'$ are neighboring data sets in the invariant data universe. Define $\|\bx\|_{\bH} = \sqrt{\bx^T\bH\bx}$ to be the $\bH-$weighted norm of a vector $\bx$ where $\bH$ is a positive definite matrix. Then
\[ a_M =  \sum_{i=1}^N \underset{\mathcal{D} \sim \mathcal{D}'}{\sup} \;\|Q_i(\mathcal{D})-Q_i(\mathcal{D}')\|_{\bH_{n_i}}.\]
Thus, the proposed procedure would not require any optimization for implementation of Gaussian -DP mechanism under simple aggregation constraints. This is a distinct advantage of the proposed method. 

\begin{Remark}
For generating under constraints, one could also view this conditionally on the query. For example, in order to generate a Gaussian noise vector $\bX$ under the condition $\bba^{\prime}\bX = \bzero,$ one could generate from the conditional distribution $\bX | \bba^{\prime} \bX = \bzero.$ Under the Gaussian mechanism with constant noise variance, the conditional distribution is a simple multivariate Gaussian, and it is no surprise that the correlation matrix coincides with the intra-class correlation matrix obtained from the prescribed algorithm. When the privacy budgets are different for different records, however, it is challenging to devise a conditional algorithm where the marginal variances of the conditional noise distribution are specified.
\end{Remark}

The implementation can be adapted to the nested (hierarchical) structure that is common in Census-like application. The procedure has the advantage that it naturally maintain the parent-child consistency across all branches of the hierarchy while honoring the invariant constraints at the highest level of the hierarchy.  Suppose now we have a database $\mathcal{D}$ which is partitioned into a nested structure. At the top level $\mathcal{D}$ is split into states, then counties, then blocks. For simplicity we consider only three levels in this simulation, but the same idea is easily extensible to deeper hierarchical structures.

Denote the state, county and block level counts by $S_{\text{state},i}, S_{\text{county},j}, $ and $S_{\text{block},k}$. We require that 
\begin{align*}
    S_{\text{state},i} &= \sum_{j \in I_{\text{county},i}} S_{\text{county},j} \hspace{15pt}\forall i \in I_{\text{state}} \\
    S_{\text{county},j}&=\sum_{k \in I_{\text{block},j} }S_{\text{block},k}\hspace{15pt} \forall j \in I_{\text{county},i}\forall i \in I_{\text{state}}
\end{align*}

Here we have written $I_{\text{block},j}$ for the set of blocks in the $j$th county, $I_{\text{county},i}$ for the set of counties in the $i$th state and $I_{\text{state}}$ for the set of all states. 

One  may have a different privacy budget at each level. The state level counts $S_{\text{state},i}$ are kept as invariants so must be undisturbed by the privacy mechanism. The implementation starts by applying Gaussian noise at the block level, according to the constraint that the total state counts are invariant. Thus, we replace $$ S_{\text{block},k} \rightarrow S_{\text{block},k} + X_{k},$$
where $X_{k}$ are jointly normal and satisfy $\sum_{j \in I_{\text{county},i}}\sum_{k\in I_{\text{block},j}}X_k=0$ almost surely for each state $i$. The covariance matrix of $\mathbf{X} = (X_k)$ will be assumed to have the form $\boldsymbol{\Sigma}_{\text{block}} = \sigma_{\text{block}}^2 \mathbf{R}_{\text{block}}$, where $\mathbf{R}_{\text{block}}$ is chosen to satisfy the constraint, and where $ \sigma_{\text{block}}^2 $ will be determined later. Since we have added noise at the block level, consistency demands that at the county level we have new totals $$ S_{\text{county},j} \rightarrow  S_{\text{county},j}  + \sum_{k\in I_{\text{block},j}} X_k.$$ Write $\tilde{\mathbf{X}} = (\sum_{k\in I_{\text{block},j}} X_k)_{j}$ for the implied added noise. Define:
\begin{align*}
    \boldsymbol{R}_\text{reduced} = \sum_{k\in I_{\text{block},j}}\sum_{k'\in I_{\text{block},j'}} (\mathbf{R}_\text{block})_{k,k'},
\end{align*} 
This may not be a true correlation matrix because the diagonal entries may vary. However, this is not an issue because Theorem \ref{thm:gaussianDP} only uses the positive semi-definiteness of $\mathbf{R}$. Now let $\sigma_{\text{min}}^2(\mathbf{R},\mathbf{M},\epsilon,\delta)$ be the minimum required noise variance according to Theorem \ref{thm:gaussianDP}. If we choose 
$$ \sigma_\text{block}^2 = \max \{\sigma_{\text{min}}^2(\mathbf{R}_\text{block},\mathbf{M}_\text{block}, \epsilon_\text{block},\delta_\text{block}),\sigma_{\text{min}}^2(\mathbf{R}_\text{reduced},\mathbf{M}_\text{county}, \epsilon_\text{county},\delta_\text{county})\},$$
we will have the required privacy  guarantee at both levels.

Table~1-2 show the results of a limited synthetic experiment that demonstrates this simple case of noise addition at multiple levels. Simulated population counts were generated by randomly generating a number of counties for each of the two states, then randomly generating a number of blocks for each county. County level population counts were generated randomly and totals were calculated for the county and state levels. A block-level Gaussian noise vector was generated according to the procedure above with $\sigma_{block} = 46.041$. The noise was rounded to the nearest integer and applied at the block level. County-level noise was calculated using the un-rounded block level noise by summing according to the consistency requirement and then rounding. Rounding to the nearest integer does not break privacy guarantees but has a chance of introducing small deviations from consistency. To maintain strict consistency, one can omit this step and allow non-integer counts in the noise-infused data set.

\begin{table}[t]\centering\begin{tabular}{|c|c|c|c|c|c|c|c|c|c|c|c|c|c|}\hline
State & \multicolumn{6}{c|}{2576} &\multicolumn{7}{c|}{2674} \\\hline
County & \multicolumn{3}{c|}{1239} &\multicolumn{3}{c|}{1337} &\multicolumn{2}{c|}{546} &\multicolumn{2}{c|}{461} &\multicolumn{3}{c|}{1667} \\\hline
Block & 861&271&107&701&21&615&215&331&88&373&872&664&131\\\hline
\end{tabular}\caption{The initial population counts are shown for three levels: state, county and block. The counts are set up so that the count sums are consistent between levels.  }\label{table:tab}\end{table}
\vspace{5pt}
\begin{table}[t]\centering\begin{tabular}{|c|c|c|c|c|c|c|c|c|c|c|c|c|c|}\hline
State & \multicolumn{6}{c|}{2576} &\multicolumn{7}{c|}{2674} \\\hline
County & \multicolumn{3}{c|}{1270} &\multicolumn{3}{c|}{1306} &\multicolumn{2}{c|}{618} &\multicolumn{2}{c|}{428} &\multicolumn{3}{c|}{1628} \\\hline
Block & 850&297&123&699&30&577&294&324&45&383&854&735&39\\\hline
\end{tabular}\caption{Gaussian noise with  is added to the simulated county level counts, using the procedure described in Section 5.1. The noise values are rounded to the nearest integer before being added to the counts. Though the state level totals remain unchanged, the lower level counts are obfuscated to meet the privacy requirement. }\label{table:tab2}\end{table}

\subsection{Laplace Mechanism With Aggregation Constraints}

\begin{figure}[t]
\centering
  \includegraphics[width=2.5in,alt={A QQ plot showing agreement between simulated data and theoretical quantiles}]{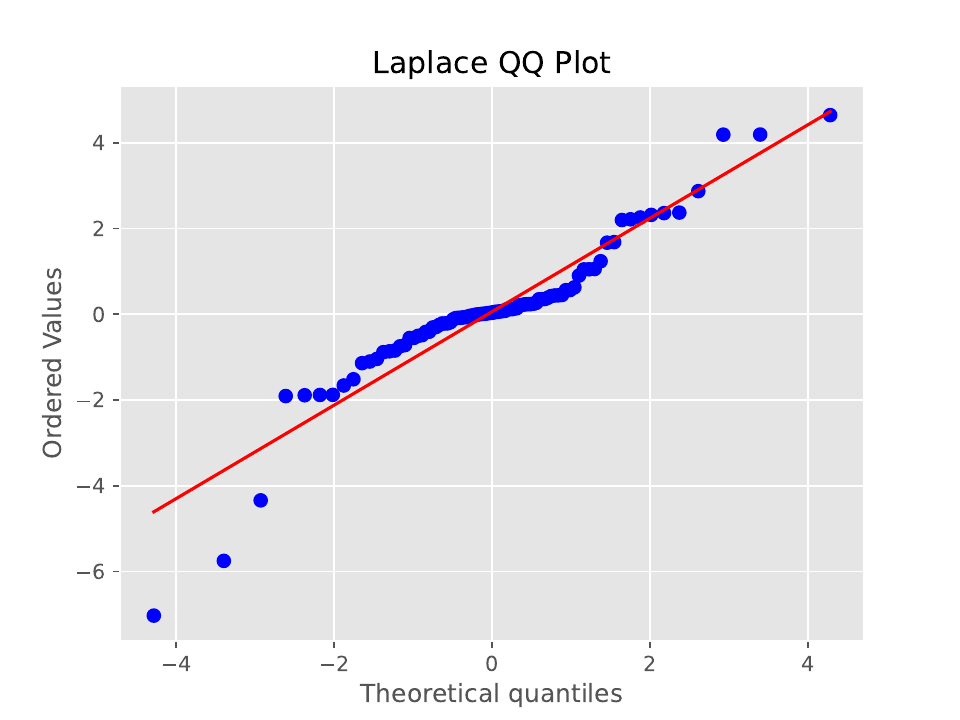}
  \includegraphics[width=2.5in,alt={A histogram of the absolute deviations from the constraint.}]{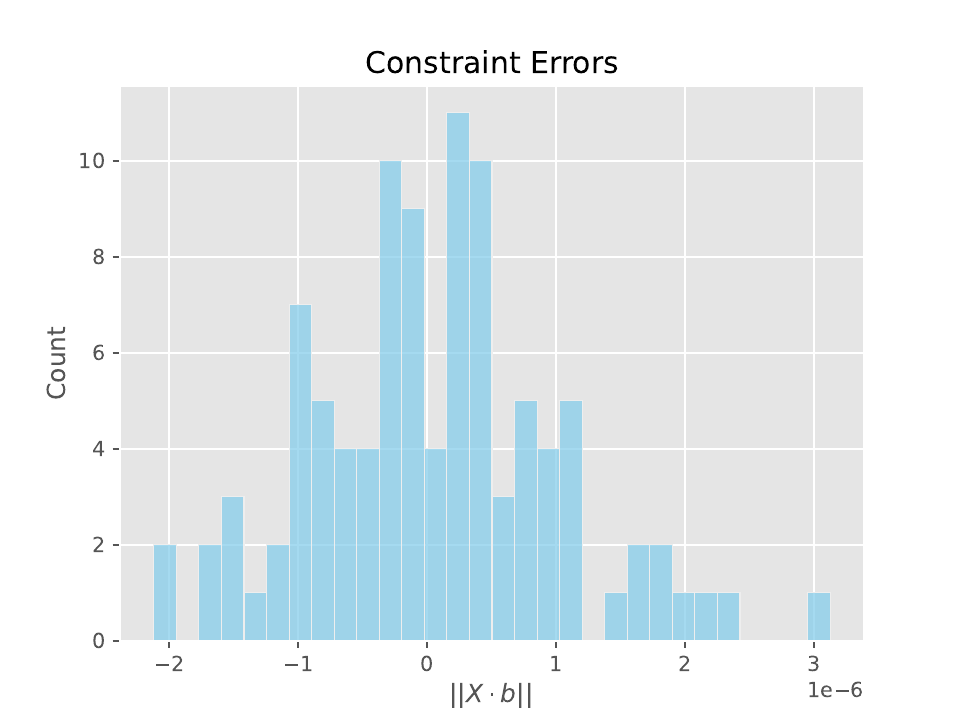}
  \caption{Left:  A Q-Q Plot shows  agreement with the desired marginal distribution; Right: Absolute deviations from satisfaction of constraint when sampling noise using the method described.}
\label{fig:qqplot}  
\end{figure}
The proposed noise addition can be easily adapted to have marginally Laplace noise by representing the Laplace distribution as a Gaussian scale mixture. Once the scales are generated from the mixing distribution, a Gaussian noise vector is generated following the proposed scheme with a desired correlation matrix.  
We considered the case where $\bB = \bba$ was a single balanced constraint. We  use the procedure in Section~2 and the construction from Theorem~3.2 of \cite{Delorme1993} to generate Laplace distributed random variates that satisfied the required constraint. 
We choose $n=100$ where $n$ is the size of the database and generated $r=100$ samples. The Q-Q plot of the sampled noise is shown in the left panel of Figure~\ref{fig:qqplot}, along with a plot of the deviations of $\bX\cdot\bba$ from $0$ in the right panel of Figure~\ref{fig:qqplot}. As expected, the deviations are negligible and are of the order of 10-E6.   

The correlation matrix is found following Algorithm~2. Figure~2 shows the convergence metrics of the algorithm both with and without the maximal entropy penalty.


\begin{figure}[h!]
  \includegraphics[width = \linewidth,alt={Three plots where the left shows the convergence of the loss function, the middle shows the penalty term and the right shows the trace term.}]{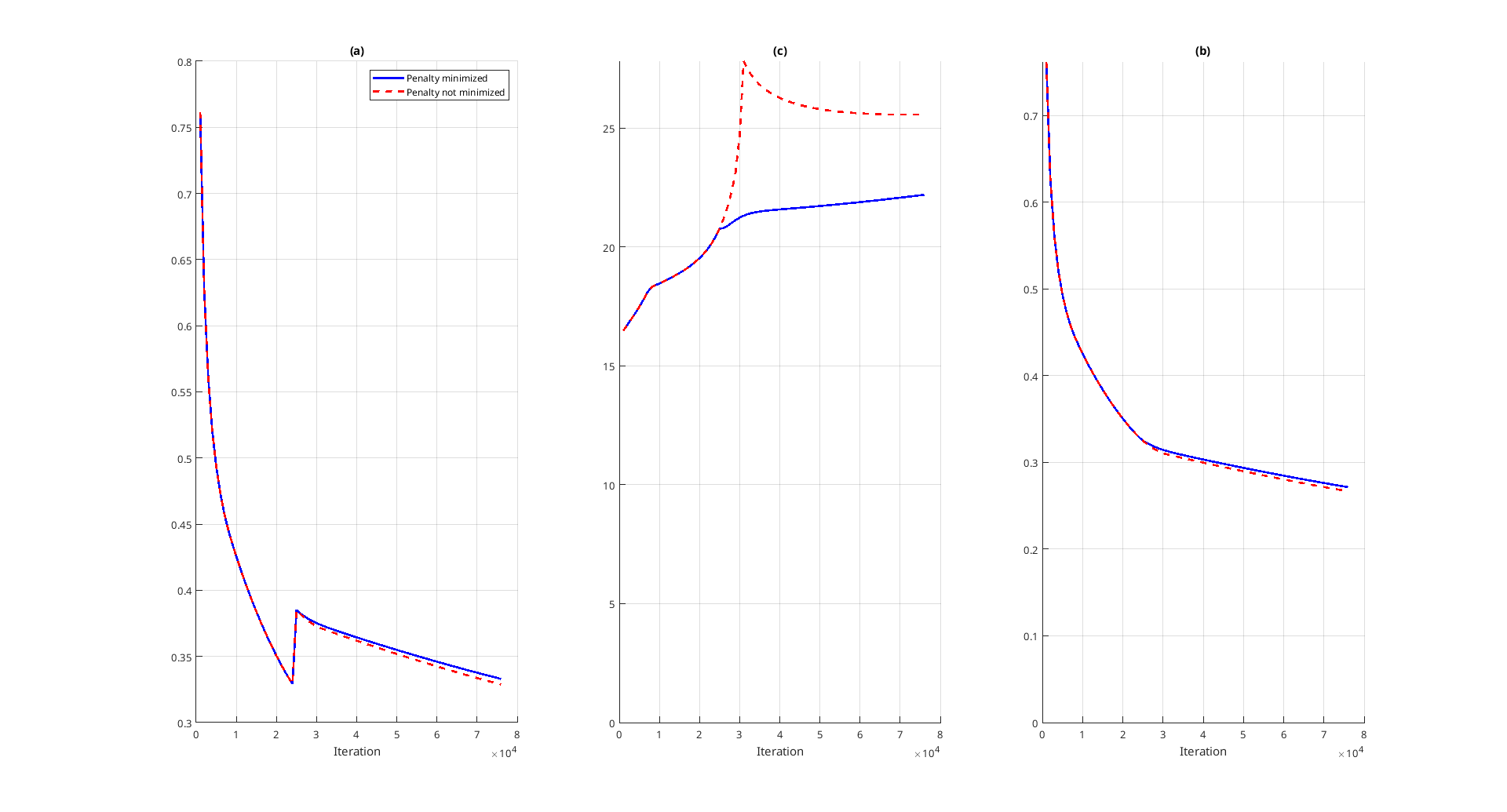}
  \caption{Algorithm 1 Convergence on simulated data. The blue color indicates the penalty term $\log \det (\mathbf{K R K^T})$ is included in the minimization. (a) The loss $\text{Tr}(\mathbf{RB^T B}) - \tau \log \det (\mathbf{K R K^T})$ (b) The penalty term $- \log \det (\mathbf{K R K^T})$
  \label{fig:alg1convergence} (c) The trace term $\text{Tr}(\mathbf{RB^T B}) $}
\end{figure}

\section{Discussion}
We have developed a differentially private data release mechanism that explicitly preserves linear aggregation constraints while providing formal theoretical privacy guarantees. Unlike standard noise-addition mechanisms, which may violate known aggregate relationships, the proposed approach ensures that released data remain consistent with prescribed linear constraints. This property is particularly important in applications where aggregate quantities are known exactly or must be preserved for legal, operational, or scientific reasons. We established the differential privacy guarantees of the proposed mechanism and, in the course of the theoretical analysis, derived some new and potentially useful properties of the associated correlation matrices. These results are of independent mathematical interest.

The present work focuses on linear aggregation constraints. Future investigation will look into more general systems of linear equality and inequality constraints. Such constraints arise naturally in many statistical and data publication problems, particularly when unit-level quantities are required to satisfy nonnegativity or other structural restrictions. One can choose to sample repeatedly from the proposed  mechanism until a positive  value is realized, thereby truncating the generation to the restricted set. Since information about the range restriction is not private, post-processing truncation would not necessarily change the privacy guarantees. However, it is obvious that the procedure will no longer satisfy the marginal distribution constraint {\bf [R2]} that we had set out as one of the requirements. 

Developing differentially private mechanisms that satisfy conic constraints while maintaining strong privacy and utility guarantees presents both theoretical and computational challenges. In particular, the geometry of the resulting feasible region, which is typically a convex polyhedron or polyhedral cone, raises interesting questions concerning the design of optimal noise distributions, efficient projection algorithms, and statistical inference under constrained perturbations.


\section*{Acknowledgments}
The authors are grateful to the late Professor Nicholas J. Higham for generously providing the code for the 'nearest correlation algorithm' from his work which the authors used to compare the correlation matrix generation scheme. The authors are also thankful for Professor Higham's helpful correspondence regarding the use of the nearest correlation code.



\section*{Appendix: Proofs of Theoretical Results}

\noindent{\bf Proof of Theorem~\ref{thm:gaussianDP}}

\begin{proof}
We start with the case where $\boldsymbol{\Sigma}$ is non-singular. Let $\mathbf{v} = Q(\mcD) - Q(\mcD')$. The privacy loss for a multivariate Gaussian  mechanism is 

\[
\mathcal{L} = -\frac{1}2 \mathbf{X} ^T \boldsymbol{\Sigma}^{-1}\mathbf{X} + \frac{1}2(\mathbf{X} + \mathbf{v})^T \boldsymbol{\Sigma}^{-1}( \mathbf{X} + \mathbf{v}) \\
 = \mathbf{v}^T \boldsymbol{\Sigma}^{-1} \mathbf{X} + \frac{1}2\mathbf{v}^T \boldsymbol{\Sigma}^{-1}\mathbf{v} \\
\]

Let $a = \mathbf{v}^T \mathbf{R}^{-1} \mathbf{v}$ where $\boldsymbol{\Sigma} = \sigma^2 \mathbf{R}$. For a given value of $\mathbf{v}$, we have $\mathcal{L} \sim N(\frac{a}{2\sigma^2} , \frac{a}{\sigma^2}).$

Thus we have 
\[
 P(\mc{L} \geq \epsilon)   = 1 - \Phi\left(\frac{\mathcal{\epsilon}-\frac{1}2 \frac{a}{\sigma^2}}{\sqrt{\frac{a}{\sigma^2}}}\right)\\
     =1 - \Phi\left(\frac{\epsilon\sigma}{\sqrt{a}} - \frac{\sqrt{a}}{2\sigma}\right)
\]
For any $a>0, \epsilon > 0$ define $t_{a, \epsilon}(\sigma) = \frac{\epsilon \sigma }{\sqrt{a}} - \frac{\sqrt{a}}{2\sigma}$.
Differentiating $t_{a, \epsilon}(\sigma)$  with respect to $a$ gives $$\frac{\partial}{\partial a}t_{a, \epsilon}(\sigma)) = -\frac{\epsilon\sigma}{2a^{\frac{3}2}} - \frac{1}{4\sigma\sqrt{a}} < 0.$$ 
Define $a_0 = \max\{\mathbf{v}^T \mathbf{R}^{-1} \mathbf{v} | \mathbf{v} = Q(\mcD) - Q(\mcD'), \mcD \sim \mcD^{\prime}\}$.
Since $\Phi$ is strictly increasing, we have for any $a, a_0$ with $a < a_0$,
$$1- \Phi(t(a)) \le 1 - \Phi(t(a_0)).$$  thus, if $1-\Phi(t(a_0)) \le \delta$, then we will have $1-\Phi(t(a)) \le \delta$ for any two neighboring data sets. Define $\kappa(t) = \frac{1}{t\sqrt{2\pi}} e^{-\frac{t^2}{2}}$.
By Mill's inequality $1 - \Phi(t) \le \kappa(t).$
 We have $\kappa'(t) = \frac{1}{\sqrt{2\pi} }\left[-\frac{e^{-\frac{t^2}{2} }}{t^{2}}-{e^{-\frac{t^2}{2} }}\right] = - \frac{1}{\sqrt{2\pi} } e^{-\frac{t^2}{2} }\left[1 + \frac{1}{t^{2}}\right] < 0 $ so $\kappa$ is strictly decreasing on $(0,\infty)$. It is easily seen that $\kappa$ maps onto $(0,\infty)$ and since it is smooth and decreasing, it has a (restricted) inverse $\kappa^{-1}: (0,\infty) \rightarrow (0,\infty)$. Computationally, $\kappa^{-1}$ may easily be determined from a table. If we consider now $t = t_{a, \epsilon}(\sigma) = \frac{\epsilon\sigma}{\sqrt{a_0}} - \frac{\sqrt{a_0}}{2\sigma}$ as a function of $\sigma$, it is smooth on positive $\sigma$. Furthermore, 
$$\frac{\partial}{\partial \sigma} t_{a, \epsilon}(\sigma) = \frac{\epsilon}{\sqrt{a_0}} + \frac{\sqrt{a_0}}{2\sigma^2} > 0. $$
It is evident that $t_{a, \epsilon}(\sigma)$ is a bijection from $(0,\infty)$ to $(-\infty,\infty)$ and has a strict inverse $t_{a, \epsilon}^{-1}$. Choose $\sigma = t_{a, \epsilon}^{-1}(\kappa^{-1}(\delta))$. Then we have $$P(\mathcal{L}  \ge \epsilon) = 1 - \Phi(t_{a, \epsilon}(\sigma)) \le \delta.$$ 

Thus far we have assumed $\mathbf{\mathbf{\Sigma}}$ to be non-singular. Let us relax this assumption. Suppose the privacy mechanism satisfies a linear constraint $\mathbf{A} \in \mathbb{R}^{k\times n}$ so that 
$$\mathbf{A}\mathbf{X} = \mathbf{0}.$$ Take $\mathbf{T} = \begin{bmatrix}
    \mathbf{A} \\
    \mathbf{M}
\end{bmatrix}. $ By construction $\mathbf{T}$ is invertible. Now, 
\begin{align*}
    P(\mathcal{M}(Q(\mcD)) \in S) &= P(\mathbf{X} \in S-Q(\mcD))\\
    &=P(\mathbf{X} \in (S-Q(\mcD)) \cap \ker \mathbf{A})\\
    & = P(\mathbf{T}\mathbf{X}\in \mathbf{T}\{(S-Q(\mcD) \cap \ker \mathbf{A}\}) \\ 
    & = P\left(\begin{bmatrix}
    \mathbf{A} \mathbf{X}\\
    \mathbf{M}\mathbf{X}
\end{bmatrix}\in \left\{\begin{bmatrix}
    \mathbf{A}\mathbf{u} \\
    \mathbf{M} \mathbf{u}
\end{bmatrix} \bigg | \mathbf{u} \in (S-Q(\mcD)) \cap \ker \mathbf{A}\right\}\right) \\
    & = P\left(\begin{bmatrix}
    \mathbf{0}\\
    \mathbf{M}\mathbf{X}
\end{bmatrix}\in \left\{\begin{bmatrix}
    \mathbf{0} \\
    \mathbf{M} \mathbf{u}
\end{bmatrix} \bigg | \mathbf{u} \in (S-Q(\mcD)) \cap \ker \mathbf{A}\right\}\right) \\
&= P(\mathbf{M}\mathbf{X} \in \mathbf{M}((S-Q(\mcD))\cap \ker \mathbf{A}))\\
&= P(\mathbf{M}\mathcal{M}(Q(\mcD)) \in \mathbf{M}(Q(\mcD) + [(S-Q(\mcD))\cap \ker \mathbf{A}]))
\end{align*} 
Therefore,  $\mathbf{M}\mathcal{M}(Q(\mcD)) = \mathbf{M} Q(\mcD) + \mathbf{M} \mathbf{X}$ can be seen as a new unconstrained query and privacy mechanism $\mathcal{M}'$ with $Q'(\mcD) = \mathbf{M} Q(\mcD)$ and $\mathbf{X}' = \mathbf{M}\mathbf{X} \sim N(0,\mathbf{M}\mathbf{\Sigma}\mathbf{M}^T)$. Hence if $\mathbf{M}\mathbf{\Sigma}\mathbf{M}^T$ is non-singular, the previous considerations imply that $$P(\mathcal{M}'(Q'(\mcD)) \in S') \le e^\epsilon P(\mathcal{M}'(Q'(\mcD')) \in S') + \delta, $$ where $S' = \mathbf{M}(Q(\mcD) + [(S-Q(\mcD))\cap \ker \mathbf{A}]).$
Therefore, 
$$P(\mathcal{M}(Q(\mcD)) \in S) \le e^\epsilon P(\mathcal{M}(Q(\mcD')) \in S) + \delta$$
for $\mcD \sim \mcD^{\prime}$ where $\mcD, \mcD^{\prime} \in \mathscr{D}_{Q, \bA, \bc}.$ By construction, $col(\bR) = col(\bM)$ and $\bM$ is a full row rank semi-orthogonal matrix. Hence $\bM\bR\bM^T$ is nonsignular. It is clear from the definition of $a_M$, that for any other choice of $\bM$, say $\bM = \bP\bM$, where $\bP$ is orthogonal, the value of $a_M$ will remain unchanged. 
\end{proof}

\noindent{\bf Proof of Proposition~1}
\begin{proof}
 We want to solve the problem 
\[ P_{\mathcal{S}_{\bB}}(\bU) = \underset{\bV \in \mathcal{S}_{\bB}} \argmin \| \bU - \bV \|_F^2\]
 Let $\bK$ be a matrix whose columns form a basis of $\text{Col}(\bB)^{\perp}$ where $\text{Col}(\bB)$ is the column space of $\bB$. While there is no unique basis for  $\text{Col}(\bB)^{\perp}$, we make the following choice. Let $\bB = [\bP_1 : \bP_2]\begin{bmatrix} \bDelta \\ \bzero \end{bmatrix} \bQ$ be the singular value decomposition of $\bB$ where $\bP_2$ is an $n\times (n-k)$ semi-orthogonal matrix. Choose $\bK = \bP_2.$
 
 Then the projection will be an $n\times n$ matrix of rank at most $(n-k)$, and can be written as $\bK\bG$ for some matrix $\bG$. By symmetry, similarly the projection is $\bH\bK^{\prime}$ for some $\bH$. Hence, we assume the projection $P_{\mathcal{S}_{\bB}}(\bU) = \bK\bS\bK^{\prime}$ for some $\bS\in S^{n-k}$. 
Since $\bK$ is semi-orthogonal,  the problem reduces to 
\beqa 
P_{\mathcal{S}_{\bB}}(\bU) &=& \underset{\bQ \in S^{n-k}} \argmin \quad \| \bU -  \bK\bQ\bK^{\prime} \|_F^2 \\
&=& \underset{\bQ \in S^{n-k}} \argmin \quad  \tr (\bU -  \bK\bQ\bK^{\prime})(\bU -  \bK\bQ\bK^{\prime})^{\prime} \\
&=& \underset{\bQ \in S^{n-k}} \argmin \quad \tr(\bQ^2 - 2\bK^{\prime}\bU\bK\bQ)
\eeqa
Differentiating with respect to $Vec(\bQ)$ (symmetric case) and equating to zero, the solution is $\bQ = \bK\bU\bK^{\prime}.$  Hence $P_{\mathcal{S}_{\bB}}(\bU)  = \bK\bK^{\prime}\bU\bK\bK^{\prime} = \bPi \bU \bPi$ where $\bPi = \bK\bK^T$ is the orthogonal projection to the $(n-k)$ orthogonal complement $\text{Col}(\bB)^{\perp}$.
\end{proof}


\noindent{\bf Proof of Theorem~\ref{thm:multiple_constraint}}
\begin{proof}
Since $B_2=FB_1$, we have
\begin{equation}
\label{eq:B-factorization}
PB=
\begin{pmatrix}
I_k\\
F
\end{pmatrix}
B_1.
\end{equation}
First, suppose that $R\in\mathscr C^n$ satisfies $RB=0$.
Partition the permuted matrix as
\[
\widetilde R:=PRP^\top
=
\begin{pmatrix}
H & G\\
G^\top & C
\end{pmatrix},
\]
where $H\in\mathbb S^k$, $G\in\mathbb R^{k\times (n-k)}$,
and $C\in\mathscr S^(n-k)$.
Because $B_1$ is nonsingular,  $RB=0$ and
\eqref{eq:B-factorization} imply
\[ 0 = \widetilde R
\begin{pmatrix} I_k\\ F \end{pmatrix} =
\begin{pmatrix} H+GF\\ G^\top+CF \end{pmatrix}.
\]
Therefore,
\begin{equation}
\label{eq:block-identities}
G^\top=-CF,
\qquad
G=-F^\top C,
\qquad
H=F^\top CF.
\end{equation}
Since $C$ is a principal submatrix of
$\widetilde R\succeq0$, it follows that $C\succeq0$.
Also, permutation preserves the unit diagonal, so
$ C_{ii}=1,\; i=1,\ldots,(n-k), $ 
and, by \eqref{eq:block-identities},
$f_j^\top Cf_j = (F^\top CF)_{jj} = H_{jj} = 1,\; j=1,\ldots,k.$
Using
$\langle e_ie_i^\top,C\rangle=C_{ii},$ and $
\langle f_jf_j^\top,C\rangle=f_j^\top Cf_j,$
we conclude that $C$ satisfies every constraint in
\eqref{eq:nullspace-sdp}.
Conversely, suppose that $C$ is feasible for
\eqref{eq:nullspace-sdp}, and define
\[
\widetilde R=
\begin{pmatrix}
F^\top CF & -F^\top C\\
-CF       & C
\end{pmatrix}.
\]
The factorization
\begin{equation}
\label{eq:R-psd-factorization}
\widetilde R =
\begin{pmatrix} -F^\top\\ I_d \end{pmatrix}
C \begin{pmatrix} -F & I_d \end{pmatrix}
\end{equation}
shows that $\widetilde R\succeq0$.
The equality constraints in \eqref{eq:nullspace-sdp} give
$\operatorname{diag}(C)=\mathbf 1_{n-k},$ and $
\operatorname{diag}(F^\top CF)=\mathbf 1_k.$
Consequently, $\operatorname{diag}(\widetilde R)=\mathbf 1_n$,
and hence
$R:=P^\top\widetilde RP\in\mathscr C^n.$
 Furthermore,
\[ \begin{pmatrix} -F & I_d \end{pmatrix}
\begin{pmatrix} I_k\\ F \end{pmatrix} =0.
\]
Combining this identity with
\eqref{eq:B-factorization} and
\eqref{eq:R-psd-factorization} yields
\[ RB = P^\top\widetilde R
\begin{pmatrix} I_k\\ F \end{pmatrix} B_1 =0.
\]
Finally, \eqref{eq:block-identities} shows that every $R\in\mathscr C^n$ satisfying $RB=0$ has the form
\eqref{eq:nullspace-reconstruction}.
The corresponding $C$ is uniquely determined as the lower $(n-k)\times (n-k)$ principal block of $PRP^\top$. Thus the stated mapping is a bijection.
\end{proof}

\noindent{\bf Proof of Corollary~\ref{cor:single_constraint}}
\begin{proof}
If $\bba$ is not balanced, then we can get obvious contradiction. Assume $\bba$ is balanced. By the argumentso of  Theorem~\ref{thm:multiple_constraint}, it is enough to consider $\bba = (1, \bff^T)^T$ where $\bff = (f_1, \ldots, f_{n-1}).$ 
Note that,  by Theorem~\ref{thm:multiple_constraint} it is enough to show  existence of $\bC \in \mathscr C^{n-1}$ such that $\bff^T\bC\bff = 1.$
By Remark~\ref{rem:maxdet-choice} without loss of generality,  $|f_i| \leq 1.$
Since $\mathscr{C}^{n-1}$ is a compact convex set, the set  $\{\bC \in \mathscr{C}^{n-1}:   \bff^T\bC\bff \}$ is a closed interval. Hence it is enough to show that there exist $\bC_1, \bC_2 \in \mathscr{C}^{n-1}$ such that $\bff^T\bC_1\bff \geq 1$ and $\bff^T\bC_2 \bff \leq 1.$ Without loss of generality we can assume $1 \geq f_1 \geq \cdots \geq f_{n-1} \geq 0.$ Since $\bba$ is balanced,  we have $\sum_i f_i \geq 1$ and hence for $\bC_1 =  \bone_{n-1}\bone_{n-1}^T$ we have $\bff^T\bC_1\bff \geq 1.$ Similarly is $\bee = (1, -1, \ldots, (-1)^{n-2})^T$, then for $\bC_2 = \bee\bee^T,$ we have $\bff^T\bC_2 \bff  \leq f_1^2 \leq 1.$
\end{proof}

\noindent{\bf Proof of Theorem~\ref{thm:maxrank}}\\
We need the following lemmas in the proof of Theorem~\ref{thm:maxrank}. We state and prove the lemmas before presenting the proof of Theorem~\ref{thm:maxrank}

\begin{Lemma}
    Every point $\mathbf{x}$ in a convex set $\mathbf{C}$ belongs to the relative interior of a unique face $\mathbf{F}_x$. 
    \label{lem:uniqueface}
\end{Lemma}
\begin{proof}
    Theorem 18.2 of \cite{rockafellar-1970a} states that the relative interiors of a convex set partition it. Thus $x\in \text{ri}(\mathcal{F})$ for some face $\mathcal{F}$. Suppose $x\in \text{ri}(\mathcal{G})$ also. Since the relative interiors form a partition we have $\text{ri}(\mathcal{G}) = \text{ri}(\mathcal{F})$. Suppose $y \in \mathcal{G}$. Then every point $$z = (1-t)x + t y $$ is in $\mathcal{G}$. In particular, because $x$ is in the relative interior of $\mathcal{G}$, we can find some $\delta$ such that if $t<\delta$ then $z \in \text{ri}(\mathcal{G})= \text{ri}(\mathcal{F})\subseteq\mathcal{F}$. It follows by the definition of a face that $y \in \mathcal{F}$. Thus we have shown $\mathcal{G} \subseteq \mathcal{F}$. A similar argument shows that $\mathcal{F}\subseteq\mathcal{G}$. 
\end{proof}

\begin{Lemma}
\label{lem:smallestface}
    Let $\mathbf{X}\in \mathscr{C}_n$. Then the smallest face $\mathcal{F}_{\mathbf{X}}$ containing $\mathbf{X}$ can be expressed as $$\mathcal{F}_{\mathbf{X}} = \{\mathbf{Y} \in \mathscr{C}_n | \ker \mathbf{Y} \supseteq \ker \mathbf{X}\}.$$ That is, every face containing $\mathbf{X}$ is a superset of $\mathcal{F}_\mathbf{X}$. 
\end{Lemma}
\begin{proof}
Clearly $\mathbf{X}$ belongs to the face $\mathcal{F}_{\mathbf{X}}$. Suppose there is another face $\mathcal{G}$ containing $\mathcal{X}$. For some space $\mathcal{V}$ we have $\mathcal{G} = \{\mathbf{Y} \in \mathscr{C}_n | \ker \mathbf{Y} \supseteq \mathcal{V}\}$ so that $\ker \mathbf{X} \supseteq \mathcal{V}$. But then for every $\mathbf{Y} \in \mathcal{F}$, $\ker \mathbf{Y} \supseteq \ker \mathbf{X} \supseteq \mathcal{V}$, so $\mathbf{Y} \in \mathcal{G}$. Hence $\mathcal{F}_\mathbf{X} \subseteq \mathcal{G}$. 
\end{proof}

\begin{Lemma} 
\label{lem:riofsmallestface}
If $\mathbf{X} \in \mathscr{C}_n$ then it belongs to the relative interior of $\mathcal{F}_\mathbf{X}$ where $\mathcal{F}_\mathbf{X}$ is the smallest face of $\mathscr{C}_n$ containing $\mathbf{X}$. 
\end{Lemma}
\begin{proof}
    By Lemma \ref{lem:uniqueface}, $\mathbf{X}$ is in the relative interior of a unique face $\mathcal{H}$. By Lemma \ref{lem:smallestface}, the smallest face containing $\mathbf{X}$ is $\mathcal{F}_\mathbf{X} = \mathcal{F}_{\mathbf{X}} = \{\mathbf{Y} \in \mathscr{C}_n | \ker \mathbf{Y} \supseteq \ker \mathbf{X}\}$. 
    Suppose $\mathbf{Y} \in \mathcal{H}$ but $\mathbf{Y} \notin \mathcal{F}_\mathbf{X}$. We have for all $t\in (0,1)$ that $t\mathbf{Y} + (1-t)\mathbf{X} \notin \mathcal{F}_\mathbf{X}$. Otherwise, we would violate the definition of a face, since $\mathbf{X},\mathbf{Y} \in \mathscr{C}_n$ and $ t\mathbf{Y} + (1-t)\mathbf{X} \notin \mathcal{F}_\mathbf{X}$ but $\mathbf{Y} \notin \mathcal{F}_\mathbf{X}$. Since $\mathbf{Y} \in \mathcal{H}$ and $\mathbf{X} \in \text{ri}(\mathcal{H})$, there exists $t>0$ such that $\mathbf{X} \pm t(\mathbf{Y}-\mathbf{X}) \in \mathcal{H}$. Since $$ \frac{1}2 [\mathbf{X}+t(\mathbf{Y}-\mathbf{X})] + \left(1-\frac{1}2\right)[\mathbf{X} - t(\mathbf{X} - \mathbf{Y})] = \mathbf{X} \in \mathcal{F}_\mathbf{X},$$ it follows again by the definition of a face that $\mathbf{X} \pm t(\mathbf{Y} - \mathbf{X}) \in \mathcal{F}_\mathbf{X}$. But this is a contradiction since $\mathbf{X} + t(\mathbf{Y}-\mathbf{X}) = t\mathbf{Y} + (1-t)\mathbf{X} \notin \mathcal{F}_\mathbf{X}$. Therefore $\mathbf{Y} \in \mathcal{H}$ implies $\mathbf{Y} \in \mathcal{F}_\mathbf{X}$. Hence $\mathcal{H} \subseteq \mathcal{F}_X$. But since $\mathcal{F}_\mathbf{X}$ is the smallest face, we must have $\mathcal{H} = \mathcal{F}_X$. So $\mathbf{X} \in \text{ri} (\mathcal{F}_X)$. 
    
\end{proof}

\begin{Lemma}
\label{lem:expressionforri}
If the relative interior of the face $\mathcal{V}_\mathbf{B}$ is nonempty, it can be expressed as $$\text{ri}(\mathcal{V}_\mathbf{B}) = \{\mathbf{R} \in \mathscr{C}_n | \ker \mathbf{R} = \ker \mathbf{R}_0\}$$ for some $\mathbf{R}_0 \in \text{ri}(\mathcal{V}_\mathbf{B})$. 
    
\end{Lemma}
\begin{proof}
    For the entirety of the proof, fix any $\mathbf{R}_0 \in \text{ri}(\mathcal{V}_\mathbf{B}).$ First, suppose $\mathbf{R}\in \mathscr{C}_n$ and $\ker \mathbf{R} = \ker \mathbf{R}_0$. 
By Lemma \ref{lem:uniqueface}, $\mathbf{R}$ lies in the relative interior of a unique face $\mathcal{F}_\mathbf{R}$, which by 
Lemmas \ref{lem:smallestface} and  \ref{lem:riofsmallestface}, can be written as $$\mathcal{F}_\mathbf{R} = \{\mathbf{Y} \in \mathscr{C}_n | \ker \mathbf{Y} \supseteq \ker \mathbf{R}\}.$$ Clearly $\mathcal{F}_\mathbf{R} = \mathcal{F}_\mathbf{R_0}$ since the kernels are equal. But then by Lemma \ref{lem:uniqueface}, $$\mathbf{R} \in \text{ri}(\mathcal{F}_\mathbf{R}) = \text{ri}(\mathcal{F}_{\mathbf{R}_0}) = \text{ri}(\mathcal{V}_\mathbf{B}).$$
Now on the other hand, suppose $\mathbf{R} \in \text{ri}(\mathcal{V}_\mathbf{B})$. Since both $\mathbf{R}$ and $\mathbf{R}_0$ belong to the relative interior of a unique face, we have $\mathcal{V}_\mathbf{B} = \mathcal{F}_\mathbf{R} = \mathcal{F}_{\mathbf{R}_0}$. So $$ \mathcal{V}_\mathbf{B} = \{\mathbf{Y} \in \mathscr{C}_n | \ker \mathbf{Y} \supseteq \ker \mathbf{R}\}=\{\mathbf{Y} \in \mathscr{C}_n | \ker \mathbf{Y} \supseteq \ker \mathbf{R}_0\}.$$ We have $\mathbf{R}, \mathbf{R}_0 \in \mathcal{V}_\mathbf{B}$ so $\ker \mathbf{R} \subseteq \ker \mathbf{R}_0$ and $\ker \mathbf{R}_0 \subseteq \ker \mathbf{R}$. Thus $\ker \mathbf{R} = \ker \mathbf{R}_0$.
\end{proof}
\vskip 5pt

\noindent{Now we can prove Theorem~\ref{thm:maxrank}}
\begin{proof}
Let $\mathbf{R} \in \mathcal{V}_\mathbf{B}$ and suppose $\mathbf{R}$ has maximal rank. Choose $\mathbf{R}_0 \in 
\text{ri}(\mathcal{V}_\mathbf{B})$. We have $\text{rank}(\mathbf{R})\ge \text{rank}(\mathbf{R_0})$. Since $\mathcal{V}_\mathbf{B}$ is the smallest face containing $\mathbf{R}_0$ we have by Lemma $\ref{lem:smallestface}$ that $\ker \mathbf{R} \supseteq \ker \mathbf{R}_0$. By the rank-nullity theorem, $n  - \text{rank}(\mathbf{R}) \ge n  - \text{rank}(\mathbf{R}_0)$, or rearranging, $\text{rank}(\mathbf{R}_0) \ge \text{rank}(\mathbf{R})$. Thus the ranks of $\mathbf{R}$ and $\mathbf{R}_0$ are equal, and consequently their nullities are also equal. But since  $\ker \mathbf{R} \supseteq \ker \mathbf{R}_0$, we see that $\ker \mathbf{R} = \ker \mathbf{R}_0$ since a vector space can't strictly contain a subspace of equal dimension. By Lemma \ref{lem:expressionforri}, $\mathbf{R} \in \text{ri}(\mathcal{V}_\mathbf{B})$. 

For the other direction, suppose $\mathbf{R} \in \text{ri}(\mathcal{V}_\mathbf{B})$ and let $\mathbf{Y} \in \mathcal{V}_\mathbf{B}$. Similar to the above, we have that $\ker{\mathbf{Y}} \supseteq \ker{\mathbf{R}}$. Hence $n-\text{rank}(\mathbf{Y}) \ge n-\text{rank}(\mathbf{R})$ and so $\text{rank}(\mathbf{R}) \ge \text{rank}(\mathbf{Y})$. 
Suppose, $\epsilon>0$ be given. $\mathbf{R}$ belongs to the face $\mathcal{F} = \mathcal{V}_\mathbf{B}$ and is either in the relative interior or the relative boundary. If $\mathbf{R} \in \text{ri}(\mathcal{V}_\mathbf{B})$ then it has maximal rank  by Theorem \ref{thm:maxrank} and we can set $\mathbf{U} = \mathbf{0}$. Otherwise, let $\mathbf{R} \in \text{rb}(\mathcal{V}_\mathbf{B})$. Nonempty convex sets have nonempty relative interiors, so choose a $\mathbf{R}_0 \in \text{ri}(\mathcal{F})$. By Theorem 6.1 of \cite{rockafellar-1970a} we have $$(1-\lambda)\mathbf{R} + \lambda\mathbf{R}_0 = \mathbf{R} + \lambda(\mathbf{R} - \mathbf{R}_0) \in \text{ri}(\mathcal{F}).$$ Set $\lambda = \frac{\epsilon}{2||\mathbf{R}-\mathbf{R}_0||}$ and $\mathbf{U} = \lambda (\mathbf{R} - \mathbf{R}_0).$ Then $$||\mathbf{U}|| = \frac{1}2 \epsilon< \epsilon.$$
\end{proof}

\noindent{\bf Proof of Proposition~\ref{prop:lomax_violation}}

\begin{proof}
The proof follows directly from the proof of Proposition~1 in \cite{arendarczyk2018joint} and straightforward calculation. 
\end{proof}

\bibliographystyle{plainnat}
\bibliography{bibliography}

@article{DwMcNiSm2016,
author = {C. Dwork and F. McSherry and K. Nissim and A. Smith},
title  = {Calibrating noise to sensitivity in private data analysis},
journal  = { Journal of Privacy and Confidentiality},
volume = {7},
pages = {17--51},
year  = {2016}
}

@article{dwork2014algorithmic,
  title={The Algorithmic Foundations of Differential Privacy},
  author={Dwork, Cynthia and Roth, Aaron},
  journal={Foundations and Trends in Theoretical Computer Science},
  volume={9},
  number={3--4},
  pages={211--407},
  year={2014},
  publisher={Now Publishers},
  doi={10.1561/0400000042},
  url={https://www.cis.upenn.edu/~aaroth/Papers/privacybook.pdf}
}

@article{Higham2002a,
title = {Computing the nearest correlation matrix- a problem from finance},
author = {Higham, Nicholas.J.},
journal = {IMA Journal of Numerical Analysis},
volume = {22},
pages = {329--343},
year = {2002}
}

@conference{Higham1989,
  author       = {Nicholas J. Higham}, 
  title        = {Matrix nearness problems and applications},
  booktitle    = {Applications of Matrix Theory},
  year         = {1989},
  editor       = {M. J. C. Gover and S. Barnett},
  pages        = {1--27},
  publisher    = {Oxford University Press},
}

@book{Higham2002b,
  author    = {Nicholas J. Higham}, 
  title     = {Accuracy and Stability of Numerical Algorithms},
  publisher = {Society for Industrial and Applied Mathematics},
  year      = 2002,
  address   = {Philadelphia, PA, USA},
  edition   = 2,
  isbn      = {0-89871-521-0}
}

@article{Higham2016,
  author  = {Nicholas J. Higham and Natasa Strabi\'{c} }, 
  title   = {Anderson acceleration of the alternating
projections method for computing the nearest correlation matrix},
  journal = {Numer. Algorithms},
  year    = 2016,
  number  = 4,
  pages   = {1021--1042},
  volume  = 72
}

@article{Delorme1993,
  author  = {C. Delorme and S. Poljak}, 
  title   = {Combinatorial properties and the complexity of a max-cut approximation},
  journal = {European J. Combin.},
  year    = 1993,
  pages   = {313--333},
  volume  = 14
}

@article{Barrett2003,
  author  = {Wayne Barrett and Stephen Pierc\'{e}}, 
  title   = {Null spaces of correlation matrices},
  journal = {Linear Algebra and its Applications},
  year    = 2003,
  pages   = {129--157},
  volume  = 368
}

@article{Laurent1996OnTF,
  title={On the Facial Structure of the Set of Correlation Matrices},
  author={Monique Laurent and Svatopluk Poljak},
  journal={SIAM J. Matrix Anal. Appl.},
  year={1996},
  volume={17},
  pages={530-547},
  url={https://api.semanticscholar.org/CorpusID:14341215}
}

@article{censorzaknoon2018,
  author  = {Censor, Yair and Zaknoon, Maroun},
  title   = {Algorithms and Convergence Results of
             Projection Methods for Inconsistent
             Feasibility Problems: A Review},
  journal = {Pure and Applied Functional Analysis},
  year    = {2018},
  volume  = {3},
  number  = {4},
  pages   = {565--586}
}

@article{bauschke1997method,
  title={The method of cyclic projections for closed convex sets in Hilbert space},
  author={Bauschke, Heinz H and Borwein, Jonathan M and Lewis, Adrian S},
  journal={Contemporary Mathematics},
  volume={204},
  pages={1--38},
  year={1997},
  publisher={Providence, RI: American Mathematical Society}
}

@article{Cong2017,
author = {Cong, Y. and Chen, B. and Zhou, M.}, title = {Fast simulation of hyperplane-truncated multivariate normal distributions}, 
journal = {Bayesian Analysis}, 
volume = {12},
pages = {1017--1037},
year = {2017}
}

@article{BoyleDykstra1986,
author = {Boyle, J. P. and Dykstra, R. L.}, 
title = {A Method for Finding Projections Onto the Intersection of Convex Sets in Hilbert Spaces}, 
journal = {Advances in Order Restricted Statistical Inference}, 
editor = {R. Dykstra et al. (eds.)},
publisher = {Springer-Verlag Berlin Heidelberg},
pages = {28--47},
year = {1986}
}

@article{Maatouk2022,
  author  = {Hassan Maatouk, Didier Rulli\'{e}re, Xavier Bay}, 
  title   = {Sampling large hyperplane-truncated multivariate
normal distributions},
journal = {Statistics \& Probability Letters},
year = {2022}
}

@article{Hoffman1991,
author = {Y. Hoffman and E. Ribak.}, 
title = {Constrained realizations of Gaussian fields:
A simple algorithm}, 
journal = {The Astrophysical Journal}, 
volume = {380},
pages = {L5},
year = {1991}
}

@book{Journel1991,
author = {A.G. Journel and C.J. Huijbregts}, 
title = {Mining geostatistics}, 
publisher = {Academic Press}, 
year = {1976}
}

@article{Geweke1998,
author = {Geweke, John},
year = {1998},
month = {03},
pages = {},
title = {Efficient Simulation from the Multivariate Normal and Student-t Distributions Subject to Linear Constraints and the Evaluation of Constraint Probabilities},
volume = {23},
journal = {Comput. Sci. Statist.}
}

@article{Abowd20222020,
	author = {Abowd, John and Ashmead, Robert and Cumings-Menon, Ryan and Garfinkel, Simson and Heineck, Micah and Heiss, Christine and Johns, Robert and Kifer, Daniel and Leclerc, Philip and Machanavajjhala, Ashwin and Moran, Brett and Sexton, William and Spence, Matthew and Zhuravlev, Pavel},
	journal = {Harvard Data Science Review},
	number = {Special Issue 2},
	year = {2022},
	title = {The 2020 {Census} {Disclosure} {Avoidance} {System} {TopDown} {Algorithm}},
	volume = { },
}

@book{rockafellar-1970a,
  address = {Princeton, N. J.},
  author = {Rockafellar, R. Tyrrell},
  callnumber = {UniM Maths 516.08 R59},
  description = {robotica-bib},
  notes = {A SRL reference.},
  publisher = {Princeton University Press},
  series = {Princeton Mathematical Series},
  title = {Convex analysis},
  year = 1970
}

@article{kasiviswanathan2011can,
  title={What can we learn privately?},
  author={Kasiviswanathan, Shiva Prasad and Lee, Homin K and Nissim, Kobbi and Raskhodnikova, Sofya and Smith, Adam},
  journal={SIAM Journal on Computing},
  volume={40},
  number={3},
  pages={793--826},
  year={2011},
  publisher={SIAM}
}

@article{joe1993parametric,
  title={Parametric families of multivariate distributions with given margins},
  author={Joe, Harry},
  journal={Journal of multivariate analysis},
  volume={46},
  number={2},
  pages={262--282},
  year={1993},
  publisher={Elsevier}
}

@article{koehler1995constructing,
  title={Constructing multivariate distributions with specific marginal distributions},
  author={Koehler, Kenneth J and Symanowski, James T},
  journal={Journal of multivariate analysis},
  volume={55},
  number={2},
  pages={261--282},
  year={1995},
  publisher={Elsevier}
}

@article{dukic2013minimum,
  title={Minimum correlation in construction of multivariate distributions},
  author={Dukic, Vanja M and Mari{\'c}, Nevena},
  journal={Physical Review E—Statistical, Nonlinear, and Soft Matter Physics},
  volume={87},
  number={3},
  pages={032114},
  year={2013},
  publisher={APS}
}

@article{huber2015multivariate,
  title={Multivariate distributions with fixed marginals and correlations},
  author={Huber, Mark and Mari{\'c}, Nevena},
  journal={Journal of Applied Probability},
  volume={52},
  number={2},
  pages={602--608},
  year={2015},
  publisher={Cambridge University Press}
}

@inproceedings{pichler2014entropy,
  title={Entropy for singular distributions},
  author={Pichler, Georg and Koliander, G{\"u}nther and Riegler, Erwin and Hlawatsch, Franz},
  booktitle={2014 IEEE International Symposium on Information Theory},
  pages={2484--2488},
  year={2014},
  organization={IEEE}
}

@book{von1951functional,
  title={Functional operators: The geometry of orthogonal spaces},
  author={Von Neumann, John},
  number={22},
  year={1951},
  publisher={Princeton University Press}
}

@article{von1949rings,
  title={On rings of operators. Reduction theory},
  author={Von Neumann, John},
  journal={Annals of Mathematics},
  volume={50},
  number={2},
  pages={401--485},
  year={1949},
  publisher={JSTOR}
}

@article{calamai1987projected,
  title={Projected gradient methods for linearly constrained problems},
  author={Calamai, Paul H and Mor{\'e}, Jorge J},
  journal={Mathematical programming},
  volume={39},
  number={1},
  pages={93--116},
  year={1987},
  publisher={Springer}
}

@book{polyak2021introduction,
  title={Introduction to continuous optimization},
  author={Polyak, Roman A and others},
  volume={172},
  year={2021},
  publisher={Springer}
}

@book{ruszczynski2011nonlinear,
  title={Nonlinear optimization},
  author={Ruszczynski, Andrzej},
  year={2011},
  publisher={Princeton university press}
}

@article{arendarczyk2018joint,
  title={The joint distribution of the sum and maximum of dependent Pareto risks},
  author={Arendarczyk, Marek and Kozubowski, Tomasz J and Panorska, Anna K},
  journal={Journal of Multivariate Analysis},
  volume={167},
  pages={136--156},
  year={2018}
}

@article{morrison1965some,
  title={Some statistical characteristics of a peak to average ratio},
  author={M. Morrison and F. Tobias},
  journal={Technometrics},
  volume={7},
  number={3},
  pages={379--385},
  year={1965},
  publisher={Taylor \& Francis}
}

@article{arov1960extreme,
  title={The extreme terms of a sample and their role in the sum of independent variables},
  author={D. Z. Arov and A. A. Bobrov},
  journal={Theory of Probability \& Its Applications},
  volume={5},
  number={4},
  pages={377--396},
  year={1960},
  publisher={SIAM}
}

@article{balakrishnan2013scale,
  title={Scale mixtures of Kotz--Dirichlet distributions},
  author={M. Balakrishnan and E. Hashorva},
  journal={Journal of Multivariate Analysis},
  volume={113},
  pages={48--58},
  year={2013},
  publisher={Elsevier}
}

@article{breiman1965some,
  title={On some limit theorems similar to the arc-sin law},
  author={L. Breiman},
  journal={Theory of Probability \& Its Applications},
  volume={10},
  number={2},
  pages={323--331},
  year={1965},
  publisher={SIAM}
}

@article{darling1952influence,
  title={The influence of the maximum term in the addition of independent random variables},
  author={D. A. Darling},
  journal={Transactions of the American Mathematical Society},
  volume={73},
  number={1},
  pages={95--107},
  year={1952},
  publisher={JSTOR}
}

@article{haas1992maximum,
  title={The maximum and mean of a random length sequence},
  author={P. J. Haas},
  journal={Journal of applied probability},
  volume={29},
  number={2},
  pages={460--466},
  year={1992},
  publisher={Cambridge University Press}
}

@incollection{kozubowski2010distributions,
  title={The distributions of the peak to average and peak to sum ratios under exponentiality},
  author={T. M. Kozubowski and A. K. Panorska and F. Qeadan},
  booktitle={Advances in Directional and Linear Statistics: A Festschrift for Sreenivasa Rao Jammalamadaka},
  pages={131--142},
  year={2010},
  publisher={Springer}
}

@article{qeadan2012joint,
  title={The joint distribution of the sum and the maximum of IID exponential random variables},
  author={F. Qeadan and T. M. Kozubowski and A. K. Panorska},
  journal={Communications in Statistics-Theory and Methods},
  volume={41},
  number={3},
  pages={544--569},
  year={2012},
  publisher={Taylor \& Francis}
}

@article{Bailie2026Refreshment,
	author = {Bailie, James and Gong, Ruobin and Meng, Xiao-Li},
	journal = {Harvard Data Science Review},
	number = {Special Issue 6},
	year = {2026},
	month = {feb 24},
	note = {https://hdsr.mitpress.mit.edu/pub/a1bpffpz},
	publisher = {The MIT Press},
	title = {A {Refreshment} {Stirred}, {Not} {Shaken}: Invariant-{Preserving} {Deployments} of {Differential} {Privacy} for the {U}.{S}. {Decennial} {Census}},
	volume = { },
}

@article{cardenas2026multiscale,
  author    = {Soumadeep Bhowmick and Anindya Roy and Kendall Parker and David Jonathan Sebastian Cardenas},
  title     = {A Multi-Scale Privacy Framework for Smart Meter Data},
  journal = {In Proceedings of 2026 IEEE International Conference on Communications, Control, and Computing Technologies for Smart Grids (SmartGridComm)},
  pages ={1-7},
  year      = {2026},
  publisher = {IEEE}
}

@article{saunderson2012diagonal,
  title={Diagonal and low-rank matrix decompositions, correlation matrices, and ellipsoid fitting},
  author={Saunderson, James and Chandrasekaran, Venkat and Parrilo, Pablo A and Willsky, Alan S},
  journal={SIAM Journal on Matrix Analysis and Applications},
  volume={33},
  number={4},
  pages={1395--1416},
  year={2012},
  publisher={SIAM}
}

\end{document}